\documentclass[11pt,
onecolumn]{article}
\usepackage[margin=1in]{geometry}
\usepackage[T1]{fontenc}
\usepackage{amsfonts}
\usepackage{amsmath}
\usepackage{amsthm}
\usepackage{amsmath}
\usepackage{euscript}
\usepackage{graphicx}
\usepackage[utf8]{inputenc}
\usepackage{subcaption}
\usepackage{subfiles}
\usepackage{wrapfig}
\usepackage{xspace}
\usepackage{xspace}
\usepackage{float}
\usepackage{marginnote}
\usepackage{amsmath}
\usepackage{amssymb}
\usepackage{enumerate}
\usepackage{epstopdf}
\usepackage{amsthm}
\graphicspath{ {./images/} }
\newcommand{\etal}{\textit{et al}.}
\newcommand{\mnote}{\marginnote}
\newcommand{\argmin}{\operatornamewithlimits{argmin}}
\newcommand{\argmax}{\operatornamewithlimits{argmax}}

\usepackage{graphicx}
\usepackage{amsfonts}
\usepackage{amsmath}
\usepackage{euscript}
\usepackage{ifthen}

\usepackage{color}

\usepackage{todonotes}
\usepackage{ upgreek }

\newboolean{disable-comments} 
\newboolean{disable-todos}
 
\setboolean{disable-comments}{true}
\setboolean{disable-todos}{true}

\newcommand{\tnote}[1]{\ifthenelse{\boolean{disable-comments}}{}{ \mnote{#1}}}   

\newcommand{\BigO}{\mathcal{O}}
\newcommand{\BigOT}{\tilde{\mathcal{O}}}
\newcommand{\ds}{\mathcal{D}}
\newcommand{\eps}{\varepsilon}
\newcommand{\poly}{\mathrm{poly}}
\newcommand{\ignore}[1]{}
\newcommand{\todonote}[1]{\ifthenelse{\boolean{disable-todos}}{}{ \todo[inline]{#1}}}

\newcommand{\origG}{\mathcal{G}}
\newcommand{\K}{\EuScript{K}}

\newcommand{\Gd}{\mathbb{G}}
\newcommand{\subcell}{\xi}
\newcommand{\init}{\textsc{Build}}
\newcommand{\search}{\textsc{HungarianSearch}}
\newcommand{\augment}{\textsc{Augment}}
\newcommand{\genduals}{\textsc{GenerateDuals}}

\newcommand{\X}{\mathbb{X}}
\newcommand{\B}{\mathbb{B}}
\newcommand{\A}{\mathbb{A}}

\newcommand{\polys}{\poly\{\log{n}, 1/\eps\}}
\newcommand{\expect}[1]{\mathbb{E}[#1]}
\newcommand{\prob}[1]{\mathbf{Pr}[#1]}
\newcommand{\cell}{\square}

\newcommand{\disteuc}[2]{\| #1 - #2 \|}

\newcommand{\distsq}[2]{\| #1 - #2 \|^2}

\newcommand{\subcells}[1]{\mathbb{G}[#1]}

\newcommand{\cost}[1]{w(#1)}

\newcommand{\sqcost}[1]{c(#1)}
\newcommand{\res}[1]{\mathcal{G}_{#1}}

\newcommand{\Mopt}{M_{\text{OPT}}}
\newcommand{\Mopth}{\hat{M}_{\text{OPT}}}
\newcommand{\opt}{\text{OPT}}

\newcommand{\sync}{\textsc{Sync}}

\newcommand{\construct}{\textsc{Construct}}

\newcommand{\wspd}{\mathcal{W}}

\newcommand{\entry}{\downarrow}
\newcommand{\exit}{\uparrow}

\newcommand{\aff}[1]{\EuScript{A}(#1)}

\newcommand{\affj}[2]{\EuScript{A}_{#2}(#1)}

\newtheorem{lemma}{Lemma}[section]
\newtheorem{theorem}[lemma]{Theorem}

\newtheorem{cor}[lemma]{Corollary}

\begin{document}

\begin{titlepage}

\title{An $\tilde{O}(n^{5/4})$ Time $\varepsilon$-Approximation Algorithm for RMS Matching in a Plane}
\author{
Nathaniel Lahn\thanks{Department of Computer Science, Virginia Tech. Email:
 \texttt{lahnn@vt.edu}} \and
Sharath Raghvendra\thanks{Department of Computer Science, Virginia Tech. Email:
 \texttt{sharathr@vt.edu}       
}}
\date{}
\maketitle

\begin{abstract}
The 2-Wasserstein distance (or RMS distance) is a useful measure of similarity
between probability distributions that has exciting applications in machine
learning. For discrete distributions, the problem of computing this distance
can be expressed in terms of finding a minimum-cost perfect matching on a
complete bipartite graph given by two multisets of points $A,B \subset
\mathbb{R}^2$, with $|A|=|B|=n$, where the ground distance between any two
points is the squared Euclidean distance between them. Although there is a
near-linear time relative $\varepsilon$-approximation algorithm for the case
where the ground distance is Euclidean (Sharathkumar and Agarwal, JACM 2020),
all existing relative $\varepsilon$-approximation algorithms for the RMS
distance take $\Omega(n^{3/2})$ time. This is primarily because, unlike
Euclidean distance, squared Euclidean distance is not a metric. In this paper,
for the RMS distance, we present a new $\varepsilon$-approximation algorithm
that runs in $\BigO(n^{5/4}\poly\{\log n,1/\varepsilon\})$ time.

Our algorithm is inspired by a recent approach for finding a minimum-cost
perfect matching in bipartite planar graphs (Asathulla et al., TALG 2020).
Their algorithm depends heavily on the existence of sub-linear sized vertex
separators as well as shortest path data structures that require planarity.
Surprisingly, we are able to design a similar algorithm for a complete
geometric graph that is far from planar and does not have any vertex
separators. Central components of our algorithm include a quadtree-based
distance that approximates the squared Euclidean distance and a data structure
that supports both Hungarian search and augmentation in sub-linear time.
\end{abstract}
\end{titlepage}

\section{Introduction}
Given two sets $A$ and $B$ of $n$ points in $\mathbb{R}^2$, let $\origG(A \cup B, A\times B)$ be the complete bipartite graph on $A,B$. A matching $M$ is a set of vertex-disjoint edges of $\origG$. The matching $M$ is \emph{perfect} if it has cardinality $n$. For any $p \ge 1$, the cost of an edge $(a,b)$ is simply $\|a-b\|^p$; here, $\|a-b\|$ is the Euclidean distance between $a$ and $b$. Consider the problem of computing a matching $M$ that minimizes the sum of all its edges' costs, i.e., the matching with smallest $w_p(M)=\sum_{(a,b)\in M} \|a-b\|^p$. When $p=1$, this problem is the well-known \emph{Euclidean bipartite matching problem}. When $p=2$, the matching computed minimizes the sum of the squared Euclidean distances of its edges and is referred to as the \emph{RMS matching}. For $p =\infty$, the matching computed will minimize the largest cost edge and is referred to as the \emph{Euclidean bottleneck matching}. For a parameter $\eps >0$ and $p  \ge 1$, we say that the matching $M$ is an $\eps$-approximate matching if $w_p(M)\le (1+\eps)w_p(M_{OPT})$ where $M_{OPT}$ is a matching with the smallest cost. In this paper, we consider the problem of computing an $\eps$-approximate RMS matching in the plane and present a randomized $\BigOT(n^{5/4})$ time\footnote{We use $\BigOT(\cdot)$ to hide $\mathrm{poly}\{\log n, 1/\eps\}$ factors in the complexity. } algorithm. For the remainder of the paper, we assume that $w(M) = w_2(M)$.

When $A$ and $B$ are multi-sets, the cost of the RMS matching is also known as the $2$-Wasserstein distance -- a popular measure of similarity between two discrete distributions. Wasserstein distances are very popular in machine learning applications. \ignore{For instance, $2$-Wasserstein distance has been successfully used as a similarity metric for images where each image is viewed as a $2$-dimensional discrete distribution.} For instance, $2$-Wasserstein distance has been used as a similarity metric for images using color distributions~\cite{2WassersteinUsingColors}. A $2$-dimensional grayscale image can be represented as a discrete distribution on $2$-dimensional points, and Wasserstein distance can be used to compare the similarity between such distributions in a fashion similar to~\cite{altschulerNIPS17,cuturiNIPS13,dvurechenskyICML18}. The $2$-Wasserstein distance has also been used for $2$-dimensional shape reconstruction~\cite{2dShapeReconstruction}.

Wasserstein distance is also used as an objective function for generative adversarial neural networks (GANs). GANs are used to generate fake objects, such as images, that look realistic~\cite{wgan,improvedwgan,liu2019qwgan}. Here, we have a `real' distribution $\mathcal{R}$ and a `fake' distribution $\mathcal{F}$. Sampling $m$ images from both $\mathcal{F}$ and $\mathcal{R}$ and computing the Wasserstein distance between the two samples gives a measure of how good the fake image generator imitates real data. The matchings (or maps) corresponding to the $2$-Wasserstein distance are also attractive because they permit a unique interpolation between the distributions; see for instance~\cite{solomon2015convolutional}.

\paragraph{Previous Results.} For any weighted bipartite graph with $m$ edges and $n$ vertices, the fundamental Hungarian algorithm can be used to find a minimum-cost maximum-cardinality matching in $\BigO(mn+n^2 \log n)$ time \cite{hungarian_56}.\footnote{Note that $m=\BigO(n^2)$ in our setting.} When edge costs are positive integers upper-bounded by a value $C$, the algorithm given by Gabow and Tarjan  computes a minimum-cost maximum cardinality matching in $\BigO(m\sqrt{n}\log (nC))$ time. These combinatorial algorithms execute $\BigO(\sqrt{n})$ phases where each phase executes an $\BigO(m)$ time search on a graph to compute a set of augmenting paths. 

In geometric settings, one can use a dynamic weighted nearest neighbor data structure to efficiently execute the search for an augmenting path in $\BigOT(n)$ time. Consequently, there are many $\BigOT(n^{3/2})$ time exact and approximation algorithms for computing matchings in geometric settings~\cite{efrat_algo,argawal_phillips_rms,s_socg13,v_focs98}. Improving upon the execution time of $\Omega(n^{3/2})$ for exact and approximation algorithms remains a major open question in computational geometry. There are no known exact geometric matching algorithms for $2$-dimensions or beyond that break the $\Omega(n^{3/2})$ barrier. However, there has been some progress for approximation algorithms for $p=1$, which we summarize next.

For the Euclidean bipartite matching problem, Agarwal and Varadarajan~\cite{av_scg04} gave an $\BigO(\log(1/\eps))$ approximation algorithm that executes in $\BigO(n^{1+\eps})$ time. Indyk~\cite{i_soda07} extended this approach to obtain a constant approximation algorithm that runs in near-linear time. Sharathkumar and Agarwal~\cite{sa_stoc12} presented a near-linear time $\eps$-approximation algorithm for the Euclidean bipartite matching problem. Each of these algorithms rely on approximating the Euclidean distance by using a ``randomly shifted'' quadtree. Extending this to $p > 1$ seems very challenging since the expected error introduced by the randomness grows very rapidly when $p = 2$ and beyond. 

When the costs satisfy metric properties, the uncapacitated minimum-cost flow between multiple sources and sinks is the same as  minimum-cost matching problem. Using a generalized preconditioning framework, Sherman~\cite{sherman} provided an $\BigO(m^{1+o(1)})$ time approximation algorithm to compute the uncapacitated minimum-cost flow in any weighted graph $G$ with $m$ edges and $n$ vertices, where the cost between any two vertices is the shortest path cost between them in $G$. Using this algorithm, one can use Euclidean spanners of small size to obtain an $\BigO(n^{1+o(1)})$ time algorithm that $\eps$-approximates the Euclidean bipartite matching cost. Khesin~\etal~\cite{preconditioningTransport} provided a more problem-specific preconditioning algorithm that returns an $\eps$-approximate Euclidean bipartite matching with an improved execution time of $\BigOT(n)$. Unlike with Euclidean costs, the squared Euclidean costs do not satisfy triangle inequality, and the reduction to uncapacitated minimum-cost flow does not apply. Furthermore, there are no known spanners of small size for squared Euclidean costs.  Therefore, these previous techniques seem to have limited applicability in the context of RMS matching.

Recently, Asathulla~\etal~\cite{soda-18} as well as Lahn and Raghvendra~\cite{lr_socg19,lr_soda19} presented algorithms that exploit sub-linear sized graph separators to obtain faster algorithms for minimum-cost matching as well as maximum cardinality matching on bipartite graphs. For instance, for any bipartite graph with $m$ edges and $n$ vertices and with a balanced vertex separator of size $n^{\delta}$, for $1/2\le \delta < 1$, Lahn and Raghvendra~\cite{lr_socg19} presented a $\BigOT(mn^{\delta/(1+\delta)})$ time algorithm to compute a maximum cardinality matching.  The $\eps$-approximate bottleneck matching problem can be reduced to finding a maximum cardinality matching in a grid-based graph.  Using the fact that a $d$-dimensional grid has a balanced, efficiently computable vertex separator of size $\BigO(n^{1-1/d})$, they obtain an $\BigOT(n^{1+\frac{d-1}{2d-1}})$ time algorithm to compute an $\eps$-approximate bottleneck matching of two sets of $d$ dimensional points. 

Given the wide applicability of Wasserstein distances, machine learning researchers have designed algorithms that compute an approximate matching within an additive error of $\eps n$. Some of these algorithms run in $\BigOT(n^2C/\eps)$ for arbitrary costs~\cite{our-neurips-2019-otapprox,quanrudSOSA19}; recollect that $C$ is the diameter of the input point set. For $2$-Wasserstein distance, such a matching can be computed in time that is near-linear in $n$ and $C/\eps$. Some of the exact and relative approximation algorithms~\cite{sa_soda12} have informed the design of fast methods for machine learning applications~\cite{our-neurips-2019-otapprox}.

\paragraph{Our Results:}

Our main result is the following.
\begin{theorem}
\label{theorem:main}
For any point sets $A,B \subset \mathbb{R}^2$, with $|A|=|B|=n$, and for any parameter $0 < \eps \leq 1$, an $\eps$-approximate RMS matching can be computed in $\BigO(n^{5/4}\poly\{\log{n}, 1/\eps\})$ time with high probability.
\end{theorem}
All previous algorithms that compute an $\eps$-approximate RMS matching take $\Omega(n^{3/2})$ time. 

\paragraph{Basics of Matching:}

Given a matching $M$, an \emph{alternating path} is a path whose edges alternate between edges of $M$ and edges not in $M$. A vertex is \emph{free} if it is not matched in $M$. An \emph{augmenting path} is an alternating path that begins and ends at a free vertex. Given an augmenting path $P$, it is possible to obtain a new matching $M' \leftarrow M \oplus P$ of one higher cardinality by \emph{augmenting} along $P$. 

Standard algorithms for minimum-cost bipartite matching use a \emph{primal-dual} approach where in addition to a matching $M$, the algorithm also maintains a set of \emph{dual weights} $y(\cdot)$ on the vertices.  A  matching $M$ along with a set of dual weights $y(\cdot)$ is \emph{feasible} if, for every edge $(a,b)$, in the input graph:
\begin{align*}
    y(a) + y(b) &\leq c(a,b). \\
    y(a) + y(b) &= c(a,b) \text{\quad if } (a, b) \in M.
\end{align*}
Here, $c(a,b)$ is the cost of the edge $(a,b)$. It can be shown that any feasible perfect matching is also a minimum-cost perfect matching.

The \emph{slack} of any edge with respect to these feasibility conditions is given by $s(a,b) = c(a,b) - y(a) - y(b)$. A set of edges is \emph{admissible} if it has zero slack. The fundamental Hungarian algorithm~\cite{hungarian_56} computes a minimum-cost matching by iteratively adjusting the dual weights and finding an augmenting path $P$ containing zero slack edges. Augmenting along this admissible path does not violate feasibility. As a result, the Hungarian algorithm arrives at an optimal matching in $n$ iterations.
\section{Overview of our Approach}
Our algorithm draws insight from a recent  $\BigOT(n^{4/3})$ time algorithm for computing a minimum-cost perfect matching in bipartite planar graphs~\cite{soda-18}. The algorithm of~\cite{soda-18} relies on the existence of a planar vertex separator of size $\BigO(\sqrt{n})$. A complete bipartite graph is far from planar and does not have any vertex separators. Despite this, we are able to adapt the approach of~\cite{soda-18} to our setting. We begin with a summary of their algorithm.  
 
\paragraph{Planar Bipartite Matching Algorithm:} The algorithm of~\cite{soda-18} is a primal-dual algorithm that iteratively adjusts the dual weights of the vertices to find an augmenting path containing zero `slack' edges and then augments the matching along this path. For a parameter $r >0$, their algorithm conducts an $\BigO(n\sqrt{r})$ time pre-processing step and computes a matching of size $n-\BigO(n/\sqrt{r})$. After this, their algorithm finds the remaining augmenting paths in sub-linear time by the use of an $r$-division: An $r$-division divides any  planar graph into $\BigO(n/r)$ edge-disjoint pieces, each of size $\BigO(r)$, with only $\BigO(n/\sqrt{r})$ many \emph{boundary vertices} that are shared between pieces. The algorithm then conducts a search for each augmenting path as follows:
\begin{itemize}
    \item Using an $r$-division of a planar bipartite graph $G(A\cup B, E)$, the algorithm constructs a compact residual graph $\tilde{G}$ with a set $\tilde{V}$ of $\BigO(n/\sqrt{r})$ vertices -- each boundary vertex of the $r$-division is explicitly added to this vertex set. In addition, the compact graph has $\BigO(r)$ edges per piece and $\BigO(n)$ edges in total. The algorithm assigns a dual weight for every vertex of $\tilde{V}$ that satisfies a set of dual feasibility constraints on the edges of $\tilde{G}$. Interestingly, given dual weights on $\tilde{V}$ that satisfy the \emph{compressed feasibility} conditions, one can derive dual weights for $A\cup B$ satisfying the classical dual feasibility conditions, and vice versa. Therefore, instead of conducting a search on $G$, their algorithm searches for an augmenting path in the compact residual graph $\tilde{G}$.
    \item Their algorithm builds, for each piece of $G$, a data structure in $\BigOT(r)$ time (see~\cite{fr_dijkstra_06}). This data structure stores the $\BigO(r)$ edges of $\tilde{G}$ belonging to the piece and using this data structure, the algorithm conducts a primal-dual search for an augmenting path  in $\BigOT(|\tilde{V}|) = \BigOT(n/\sqrt{r})$ time. Over $\BigO(n/\sqrt{r})$ augmenting path searches, the total time taken is bounded by $\BigOT(n^2/r)$.
\end{itemize}

\ignore{After a $\BigO(n\sqrt{r})$ time pre-processing step that computes a matching of size $n-\BigO(n/\sqrt{r})$, their algorithm builds a shortest path data structure of~\cite{klein_mssp_05} for the directed residual graph within each of the $\BigO(n/r)$ pieces in $\BigO(r)$ time each. Using these data structures, Dijkstra's algorithm can be executed in time near-linear in the number of boundary vertices. They also observed that fully up-to-date dual weights are only needed on the boundary vertices; the other dual weights can be updated in a lazy fashion. Combining these observations, a primal-dual search for an augmenting path consisting of zero slack edges can be executed in $\BigOT(n/\sqrt{r})$ time. Therefore, the total search time for the remaining $\BigO(n/\sqrt{r})$ augmenting paths is only $\BigOT(n^2/r)$.
}

Augmenting along a path reverses the direction of its edges in the residual graph. Therefore, their algorithm has to re-build the shortest path data structure for every \emph{affected piece}, a piece containing at least one edge of the augmenting path. This can be done in $\BigOT(r)$ time per piece. In order to reduce the number of affected pieces, an additive cost of $\sqrt{r}$ is introduced to every edge incident on the boundary vertices. It is then shown that the total additive cost across all augmenting paths found by the algorithm cannot exceed $\BigO(n\log n)$, implying that the number of affected pieces is at most $\BigO((n/\sqrt{r})\log n)$. The  time taken to re-build the data structure for the affected pieces is $\BigOT(n/\sqrt{r} \log n)\times\BigOT(r)=\BigOT(n\sqrt{r})$. By choosing $r=n^{2/3}$, they balance the search time with the re-build time, leading to an $\BigOT(n^{4/3})$ time algorithm. 
\ignore{
It is worth noting that there is a recent improvement to the $\BigOT(n^{4/3})$ time algorithm of Asathulla~\etal\ that runs in $\BigOT(n^{6/5})$ time~\cite{lr_soda19} for planar graphs. The improvement results from techniques that allow multiple augmenting paths to be found each iteration. Unfortunately, attempting to apply these same techniques to our setting seems to offer no additional benefit over the algorithm of Asathulla~\etal, at least, without new observations.
}

The successful application of a compact residual network as well as the additive cost of $\sqrt{r}$ on the edges relies on the existence of an $r$-division in planar graphs. 
In order to extend these techniques to the geometric setting, we build upon ideas from another matching algorithm, which produces an $\eps$-approximation  for the Euclidean bipartite matching problem~\cite{sa_stoc12}. We give a brief overview of this algorithm next.
 
\paragraph{Approximate Euclidean Matching:} The algorithm of~\cite{sa_stoc12} introduces an $\eps$-approximation of the Euclidean distance based on a quad-tree $Q$. The input is transformed so that the optimal matching cost is $\BigO(n/\eps)$ and the height of the quad-tree $Q$ is $\BigO(\log n)$.  Any edge of the complete bipartite graph \emph{appears} at the least common ancestor of its endpoints in $Q$. The set of edges appearing within each quadtree square is then partitioned into $\polys$ many \emph{bundles} and all edges within the same bundle are assigned the same cost. This assigned cost is an upper bound on the actual Euclidean cost. Furthermore, the authors show that, if the quad-tree is randomly shifted, the expected cost assigned to any edge is at most $(1+\eps)$ times the Euclidean distance. Using this, the authors switch to computing a matching with respect to this new quad-tree distance.

Using the edge bundles and certain carefully pre-computed shortest paths in the residual graph, the algorithm of~\cite{sa_stoc12} stores a $\polys$ size \emph{associated graph} at each square of the quad-tree. Their algorithm iteratively finds a minimum-cost augmenting path $P$. Note that this is not done by using a primal-dual method, but by executing a Bellman-Ford search on the associated graph of each square that contains at least one point on the path $P$. Since each point of $P$ has at most $\BigO(\log n)$ ancestors and the size of the associated graph is $\polys$ within each square, the total time taken to find an augmenting path can be bounded by $\BigOT(|P|)$.  Augmenting the matching along $P$ requires the associated graph to be reconstructed for the $\BigO(\log n)$ ancestors of each of the points of $P$. This again can be done using the Bellman-Ford algorithm, resulting in a total update time of $\BigOT(|P|)$. The total length of all the augmenting paths computed by the algorithm can be shown to be $\BigOT(n\log n)$, and so the total time taken by the algorithm is near-linear in $n$.

\paragraph{Our Algorithm:}
Similar to the Euclidean case, we can transform our input so that our optimal matching cost is $\BigO(n/\eps^2)$ (see Section \ref{subsec:transform}) and store the input in a quadtree $Q$ of height $\BigO(\log{n})$. 
For the squared Euclidean distance, we combine the ideas from the two algorithms of~\cite{soda-18} and~\cite{sa_stoc12} in a non-trivial fashion. First, we note that using $\BigO(\polys)$ edge bundles leads to an explosion in the expected distortion. In order to keep the expected distortion small, we create approximately $\BigOT(2^{i/2})$ edge bundles for a square of side-length $2^i$\footnote{Throughout this paper, we set the side-length of the square to be the difference in the x-coordinate values of the the vertical boundaries, i.e., the Euclidean length of each of its four edges.}. This causes larger squares have many more bundles of edges (See Section~\ref{sec:distance}). For instance, a square of side-length $n$ can have roughly  $\sqrt{n}$ edge bundles. A useful property of this distance approximation is that any edge appearing in a square of side-length $2^i$ has a quad-tree distance value roughly between $\Omega(2^i)$ and $\BigO(2^{2i})$. This implies that all edges with a small quad-tree distance appear within edge bundles of the smaller squares. Like in the Euclidean case, we can show that our distance is an upper bound on the squared Euclidean distance. Furthermore, if $Q$ is a randomly shifted quad-tree, we can show that the expected cost of our distance is at most $(1+\eps)$ times the squared Euclidean distance.

In the squared Euclidean quad-tree distance, the number of edge bundles at each square of the quad tree is a polynomial in $n$.  Using these bundles, we define a sub-linear sized associated graph. However, unlike the algorithm of~\cite{sa_stoc12}, using the Bellman-Ford search procedure to find an augmenting path in the associated graph will lead to an $\Omega(n^{3/2})$ time algorithm. Therefore, instead of the Bellman-Ford algorithm, we employ a primal-dual approach. 

Prior to describing our algorithm and the data structure, we note that primal-dual search procedures, such as Hungarian search and our algorithm, find augmenting paths in increasing order of their ``costs''. As a result, such a search on quad-tree distances will initially involve only the edges with small quadtree distance and, as the algorithm progresses, larger quad-tree distances get involved. Therefore, the searches can initially be localized to smaller squares of the quad-tree and our algorithm only needs to build the associated graphs in the smaller squares. As the algorithm progresses, however, longer edges participate in the augmenting paths, which forces our algorithm to build associated graph data structures in larger squares, increasing the time taken to conduct a Hungarian search. We refer to these squares where the data structure is maintained as \emph{active} squares.

Now we present an overview of our algorithm and the data structure within an active square $\cell^*$ of width $2^i$. We partition $\cell^*$ into $\BigO(2^{2i/3})$ \emph{pieces} using a grid of side-length $2^j=2^{\lfloor 2i/3 \rfloor}$. Each piece is further recursively divided into four squares. The entire hierarchical structure is stored within a carefully defined \emph{active tree}. We build an associated graph $\tilde{G}$ at each node of the active tree. For the first level of the active tree, we build the associated graph as follows: The vertex set $\tilde{V}$ contains $\BigOT(2^{i/3})$ vertices per piece and $\BigOT(2^i)$ vertices in total. For pairs of vertices $u,v$ that belong to the same piece, we explicitly store an edge; we refer to these edges as \emph{internal} edges. There are $\BigOT(2^{2i/3})$ internal edges per piece and the internal edges in each piece can be constructed in $\BigOT(2^i)$ time (see Sections \ref{subsec:associatedgraphvertices}--\ref{subsec:edgesofassociatedgraph}). Similar associated graphs are also constructed for every subsequent levels of the active tree. Similar to the approximate Euclidean matching algorithm, these internal edges of the associated graph represent certain shortest paths in the residual graph. Additionally, for any pair of vertices $u,v \in \tilde{V}$, we add a \emph{bridge edge} between them with a cost that is approximately the squared Euclidean distance between the end-points. We do not store the bridge edges explicitly. Instead, we build an $\eps$-Well Separated Pair Decomposition (WSPD) of size $\BigOT(2^i)$ to store them. Therefore, the total size of the graph is restricted to $\BigOT(2^i)$ vertices and $\BigOT(2^{4i/3})$ edges.

Next, we define dual weights on every vertex of the associated graph and define compressed feasibility conditions that are satisfied by its edges (see Section \ref{subsec:compressedfeasibility}). 
Recollect that for planar graphs, compressed feasibility conditions are defined only on a single global compressed residual graph. In our case, however, residual graph is represented in a compressed fashion via a hierarchical set of associated graphs defined on every node of the active tree. It is significantly more challenging to design compressed dual feasibility conditions that allows for a sub-linear time Hungarian search procedure on such a hierarchical structure. Interestingly, one can use the feasible dual weights on the associated graph vertices to derive a set of dual weights satisfying the classical matching feasibility conditions (see Section \ref{subsec:sync}). Using compressed feasibility, we provide a quick way to conduct primal-dual searches on the associated graph resulting in a running time of $\BigOT(2^{4i/3})$ per search (see Section \ref{subsec:hungariansearch}). We show that the number of primal-dual searches on the associated graph of any active square with side-length $2^i$ is only $\BigOT(n/2^i)$ (see Section \ref{subsec:analysisofalg}). Therefore, the total time spent for all searches within active squares of side-length $2^i$ is $\BigOT(n2^{i/3})$. 

Suppose the primal-dual search at $\cell^*$ returns an admissible augmenting path. The algorithm then augments the matching along this path. Augmentation forces the algorithm to rebuild the set of internal edges within every \emph{affected piece} of the associated graph at $\cell^*$, i.e., pieces that contain at least one edge of $P$. In order to reduce the number of such updates, similar to~\cite{soda-18}, we assign an additive cost of roughly $\frac{\eps^22^{2i/3}}{\log{n}}$ to every bridge edge of the associated graph. We argue that this additional error does not increase the optimal matching cost by more than a multiplicative factor of $\eps$.

To bound the time taken to rebuild the internal edges, similar to~\cite{soda-18}, we argue that the total additive cost of the edges on the augmenting paths, computed over the entire algorithm, cannot exceed $\BigOT(n)$ (see Section \ref{subsec:analysisofalg}). Every bridge edge of the associated graph $\tilde{G}$ has an error of at least $\frac{\eps^22^{2i/3}}{\log{n}}$. Therefore, the number of times such edges participate across all augmenting paths is only $\BigOT(\frac{n}{2^{2i/3}})$. As a result, the total number of rebuilds of internal edges, for pieces of all active squares of side-length $2^i$, across the entire algorithm, is $\BigOT(n/2^{2i/3})$. Rebuilding the internal edges of one piece takes $\BigOT(2^{i})$ time (see Section \ref{subsec:construct}). Therefore, the total time spent rebuilding pieces  is $\BigOT(n2^{i/3})$, which matches the total time taken for all searches on the associated graph for layer $i$ active squares.

As the algorithm progresses, larger squares become active. When the side-length of the active square is approximately $n^{3/4}$, the time taken to execute a single search on the associated graph  becomes $\Omega(n)$. At this point, we show that there are only $\BigOT(n^{1/4})$ free vertices remaining. Each remaining free vertex can be matched by conducting an efficient Hungarian search on the original points in $\BigOT(n)$, taking $\BigOT(n^{5/4})$ time in total. The total time spent on searches and rebuilds on active squares with side-length at most $2^{(3/4)\log_2n} = n^{3/4}$ using our data structure is $\BigOT(n2^{(1/4)\log_2 n})=\BigOT(n^{5/4})$, giving a total running time of $\BigOT(n^{5/4})$.

\paragraph{Comparison with~\cite{lr_soda19}:} Following the work of Asathulla~\etal~\cite{soda-18}, using the same framework, Lahn and Raghvendra presented a faster $\BigOT(n^{6/5})$ algorithm to compute a minimum-cost perfect matching in planar graphs. Their main idea was to carefully compute multiple augmenting paths in one scan of the graph, leading to a faster convergence to the optimal matching. We would like to note that any augmenting path found in our algorithm is localized within an active square. Therefore, our algorithm identifies one augmenting path in a single access to an active square and many augmenting paths in a single access to the entire graph (spanning all the active squares). Unlike in the case of planar graphs, employing the approach of Lahn and Raghvendra~\cite{lr_soda19} in our setting does not lead to any additional advantage in terms of the convergence to a perfect matching. 

\paragraph{Extensions and Open Problems:}    
Achieving a near-linear execution time in the two-dimensional case and $o(n^{3/2})$ time algorithms for $d$-dimensions remain important open questions. Our approach can achieve this goal provided we overcome the following difficulty:  Currently, Hungarian search runs in time linear in the number of internal edges. In planar graphs, although the compressed residual graph has $n$ edges, one can use a shortest-path data structure by Fakcharoenphol and Rao~\cite{fr_dijkstra_06} to execute each Hungarian search in $\BigOT(|\tilde{V}|)= \BigOT(n^{2/3})$ time. Design of a similar data structure that conducts Hungarian search on associated graph in time $\BigOT(|\tilde{V}|)$ will lead to a near-linear time $\eps$-approximation algorithm for RMS matching in two-dimensions and an $o(n^{3/2})$ time algorithm in higher dimensions. 

\paragraph{Organization:} The remainder of the paper is organized as follows: In Section \ref{sec:distance} we describe the details of our distance function while highlighting differences from the distance function of~\cite{sa_stoc12}. In Section \ref{sec:dualfeas} we introduce a quad-tree based dual-feasibility condition that incorporates an additional additive cost on each edge. In Section \ref{sec:algorithm}, we give a detailed description of the algorithm, along with its analysis. The algorithm description assumes the existence of a data structure built on active squares. This data structure includes the compressed feasible matching as well as several procedures, such as the sub-linear time Hungarian search and augment that operate on a compressed feasible matching, and is described in detail in Section \ref{sec:ds}. 

\section{Our Distance Function}
\label{sec:distance}
\subsection{Initial Input Transformation}
\label{subsec:transform}
For the purposes of describing both the distance function and our algorithm, it is useful to make some assumptions about the input to our problem. Given any point sets $A', B'\subset \mathbb{R}^2$ of $n$ points, we generate point sets $A$ and $B$ with $n$ points such that each point of $A'$ (resp. $B'$) maps to a unique point of $A$ (resp. $B$) and:
\begin{itemize}
     \item[(A1)] Every point in $A \cup B$ has non-negative integer coordinates bounded by $\Delta =n^{\BigO(1)}$, 
    \item[(A2)] No pair of points $a,b$ where $a \in A$ and $b \in B$ are co-located, i.e., $\|a-b\| \ge 1$,
    \item[(A3)] The optimal matching of $A$ and $B$ has a cost of at most $\BigO(n/\eps^2)$, and,
    \item[(A4)] Any $\eps$-approximate matching of $A$ and $B$ corresponds to an $3\eps$-approximate matching of $A'$ and $B'$.
\end{itemize}
The details of this transformation are described in Section~\ref{sec:transform}, but the approach can be summarized as follows: First, we obtain an $n^{O(1)}$-approximation of the optimal matching in linear-time. We further refine this estimate by making $\BigO(\log{n})$ guesses of the optimal cost, and at least one guess gives a $2$-approximation. By executing our algorithm $\BigO(\log{n})$ times, one for each guess, at least one algorithm will have a $2$-approximation of the optimal matching cost. Using this refined estimate, we rescale the points such that the optimal cost becomes $\BigO(n/\eps^2)$. Finally, we show that rounding the resulting points to integers, such that no point of $A$ is co-located with a point of $B$, does not contribute too much error. As a result, in the rest of the paper, we assume that properties (A1)--(A4) hold. Given these assumptions, we can proceed with defining our distance function. Next, we describe a quad-tree based distance denoted by $d_Q(\cdot,\cdot)$ that approximates the squared Euclidean distances between points.

\subsection{Randomly Shifted Quadtree Decomposition}
Similar to~\cite{sa_stoc12}, we define our distance function based on a randomly-shifted quadtree. Without loss of generality, assume  $\Delta$ is a power of $2$.  First, we pick a pair of integers $\langle x,y \rangle$ each independently and uniformly at random from the interval $[0,\Delta]$. We define a square $G=[0,2\Delta]^2 - \langle x,y \rangle$ that contains all points of $A \cup B$. This square will form the root node of our quadtree, and each internal node of the tree is subdivided into $4$ equal-sized squares to form its children in the tree.

Specifically, for $\delta=\log_2(2\Delta)$ and a constant $c_1> 0$, we construct a quadtree $Q$ of height $\delta + 2\log(\log(\Delta)/\eps) + c_1 = \BigO(\log n)$ (from (A1)). The layers of $Q$ can be seen as a sequence of grids $\langle G_\delta,\ldots, G_0, \ldots,$\  $ G_{-2\log (\log(\Delta)/\eps)-c_1} \rangle$. The grid $G_i$ is associated with squares with side-length $2^i$ and the grid of leaf nodes $G_{-2\log (\log(\Delta)/\eps) -c_1}$ is associated with cells of width $1/2^{2\log (\log(\Delta)/\eps) + c_1}$. Although, cells of grid $G_0$ contain at most one point (or possibly multiple copies of the same point) and can be considered leaf nodes of the quadtree, it is notationally convenient to allow for us to define grids $G_i$ for all $i \geq -2\log(\log(\Delta)/\eps)-c_1$ and consider their cells to be part of the quadtree. Specifically, the additional levels help facilitate a cleaner definition of subcells (see Section~\ref{subsec:division-subcells}).  We say that a square $\square$ has a \emph{level} $i$ if $\square$ is a cell in grid $G_i$.  For any two cells $\square$ and $\square'$ , let $\ell_{\min}(\square,\square')$ (resp. $\ell_{\max}(\square,\square'))$ be the minimum (resp. maximum) distance between the boundaries of $\square$ and $\square'$, i.e., the minimum distance between any two points $u$ and $v$ where $u$ is on the boundary of $\square$ and $v$ is on the boundary of $\square'$. Next, we describe how any cell of this quadtree that has a level greater than or equal to $0$ can be divided into subcells, a concept essential to describe our distance function.

\subsubsection{Division of a Cell Into Subcells}
\label{subsec:division-subcells}
For any grid $G_i$ with $i \geq 0$, we define the \emph{minimum subcell size} to be 
$\mu_i = 2^{\lfloor i/2\rfloor - 2\log \frac{\log \Delta}{\eps} - c_1}$,
where $c_1 > 0$ is the constant used in the construction of the quad-tree. Each cell $\cell \in G_i$ is subdivided into a set of subcells, with each subcell having width at least $\mu_i$. In~\cite{sa_stoc12}, the minimum subcell size was much larger, being roughly $2^{i - \BigO(\log{\frac{\log{\Delta}}{\eps}})}$. In their case, dividing $\cell$ into subcells using a uniform grid of side-length $\mu_i$ was sufficient, resulting in $\BigO((2^i/\mu_i)^2) = \polys$ subcells. However, for squared Euclidean distances, much smaller subcells are required, and using a uniform grid would result in $\Omega(2^i)$ subcells, which is too large for our purposes. Instead, we replace the uniform grid of subcells with an \emph{exponential} grid of subcells, reducing the number of subcells to $\BigOT(2^i/\mu_i) = \BigOT(2^{i/2})$. We describe this process of forming the exponential grid next. For a visual example of the exponential grid, see Figure \ref{fig:subcells}.

For any cell $\square$ of $Q$ with a level $i \geq 0$, let $\square_1, \square_2, \square_3$ and $\square_4$ be its four children. We define \emph{subcells} of any cell $\cell$ as the leaf nodes of another quadtree $Q_{\square}$ with $\square$ as its root and its four children recursively sub-divided in $Q_\square$ as follows.   Let $u \leftarrow \square_1$, we recursively divide $u$ into four cells until:
\begin{itemize}
\item[(a)]   Either the side-length of $u$ is the minimum subcell size $\mu_i$, or
\item[(b)]  The side-length of $u$ is at most $(\eps/144) \ell_{\min}(\square_1,u)$.
\end{itemize}
Similarly, we decompose $\square_2, \square_3$ and $\square_4$ into subcells as well. Note that every cell of the quadtree $Q_{\square}$ is also a cell in the quadtree $Q$ and the leaves of $Q_\square$ (the \emph{subcells} of $\cell$) will satisfy (a) or (b).  We denote the subcells of $\square$ by $\subcells{\square}$. Note that, for any subcell $u \in \subcells{\square}$ where $u$ is a descendant of $\square_1$, the side-length of $u$ is larger than the minimum subcell size if and only if $\ell_{\min}(\square_1,u)$ is sufficiently large. i.e., as we move away from the boundary of $\square_1$, the subcell size becomes larger. Using this, in Lemma~\ref{lem:subcellcount}, we show that the total number of subcells for any cell $\square \in G_i$ is $\BigOT(\mu_i)$. For brevity, the proof of Lemma~\ref{lem:subcellcount} is included in Section~\ref{sec:ommittedproofs}, but the argument can be seen intuitively from the fact that the outermost ring of subcells along the boundary of $\cell_1$ has size $\BigOT(\mu_i)$. Furthermore, subcells increase in size as we move towards the center of $\cell_1$, implying that their count decreases geometrically.

\begin{lemma}
\label{lem:subcellcount}
For any cell $\cell$ of $Q$ with level $i$, the total number of subcells is $\BigOT(\mu_i)$. 
\end{lemma}

For some edge $(a,b) \in A \times B$, let $\square$ be the least common ancestor of $a$ and $b$ in $Q$. Suppose that $\cell \in G_i$; then we say that the edge $(a,b)$ \emph{appears at} level $i$. Note that, from (A2), all edges of $A \times B$ appear at or above level $1$. The quadtree distance between $a$ and $b$ defined in~\cite{sa_stoc12} is given by the distance between the subcells $\xi_a$ and $\xi_b$ of $\subcells{\square}$ that contain $a$ and $b$ respectively.  As a result, the set of edges that appear at layer $i$ can be represented using pairs of subcells from the set $\bigcup_{\cell' \in G_i} \subcells{\cell'}$. However, the use of all pairs of subcells is prohibitively expensive. We further reduce the number of pairs of subcells by grouping them into a Well-Separated Pair Decomposition which we describe next.

\subsubsection{Well-Separated Pair Decompositions}
In this section, we extend Well-Separated Pair Decomposition (WSPD) that is commonly defined for points to approximate distances between pairs of subcells. 
A Well-Separated Pair Decomposition (WSPD) is a commonly used tool that, given a set $P$ of $n$ points, compactly approximate all $\BigO(n^2)$ distances between points of $P$ by using a sparse set  $\mathcal{W}$ of only $\BigOT(n)$ well-separated pairs. Each pair $(S,T) \in \mathcal{W}$ consists of two subsets $S,T \subseteq P$ of points. For any pair of points $(u,v) \in P \times P$, there is a unique pair $(S,T) \in \mathcal{W}$ such that $(u,v) \in S \times T$. For each pair $(S,T)$, an arbitrary pair of representatives $s \in S$ and $t \in T$ can be chosen, and the distance between any pair $(s',t') \in S \times T$ can be approximated using the distance between the representatives $s$ and $t$. This approximation will be of good quality so long as the pair $(S,T)$ is well-separated, meaning the distance between any pair of points \emph{within} $S$ or \emph{within} $T$ is sufficiently small compared to the distance between any pair of points \emph{between} $S$ and $T$. 

For any parameter $\eps > 0$, using the construction algorithm  of~\cite{har2011geometric}, it is possible to build in $\BigOT(n\polys)$ time a WSPD of the edges of $A \times B$ where the costs of the edges belonging to any pair in the decomposition are within a factor of $(1+\eps)$ of each other. Furthermore, if the ratio of the largest edge to smallest edge cost is bounded by $n^{\BigO(1)}$, then it can be shown that every point participates in only $\polys$ pairs. Such a WSPD can be used to execute a single Hungarian search in near-linear time in the number of points. However, in order to execute a Hungarian search in sub-linear time, we must build a WSPD on the sub-linear number of subcells instead of the original points. Luckily, the algorithm of~\cite{har2011geometric} can be applied in a straightforward fashion to generate a WSPD on subcells. Next, we describe the properties of our WSPD on subcells.

\begin{figure}
    \centering
    \includegraphics[width=\textwidth,height=\textheight,keepaspectratio]{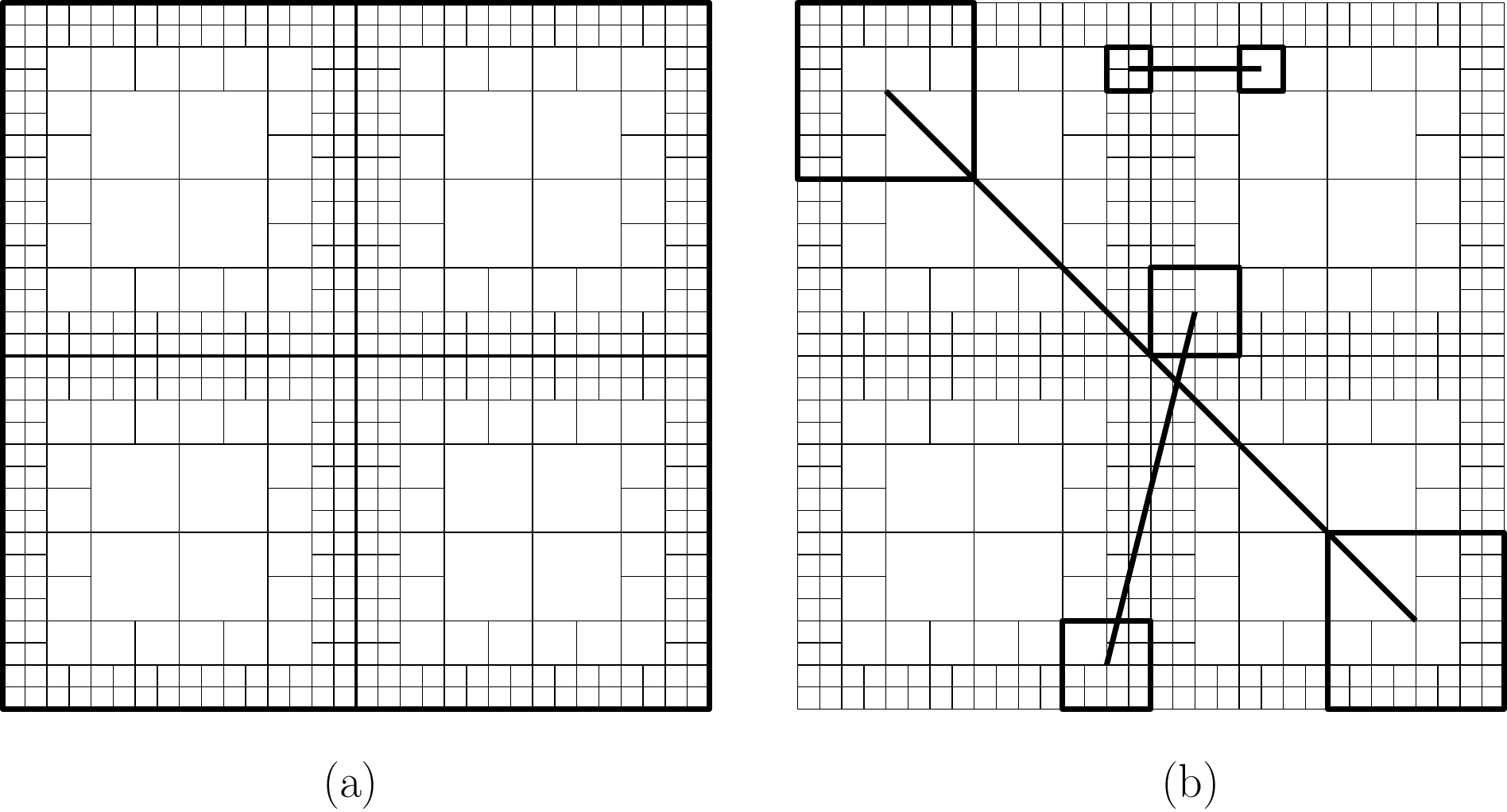}
    \caption{(a) A division of $\cell$ into subcells. (b) Examples of a few possible WSPD pairs of $\wspd_{\cell}$. Every pair of subcells in different children of $\cell$ would be represented by some pair.}
    \label{fig:subcells}
\end{figure}

For any level $i$ cell $\square$ of $Q$, consider two subsets of subcells, $S \subseteq \subcells{\square}$ and $T \subseteq \subcells{\square}$. We define $\ell_{\max}(S,T)= \max_{\xi \in S, \xi' \in T} \ell_{\max}(\xi,\xi')$. We say that $S$ and $T$ are $\eps$-well separated if, for every pair of subcells $\xi \in S$ and $ \xi'\in T$,  
\begin{equation}
\label{eq:wspdguarantee}
 \ell_{\max}(S, T) \le (1+\eps/12)\ell_{\max}(\xi,\xi').    
\end{equation}
For each cell $\square$ let $\square_1,\square_2,\square_3$ and $\square_4$ be its four children. We precompute a WSPD $\wspd_\cell=\{(S_1, T_1),$
$ \ldots, (S_r, T_r)\}$, where  $S_i \subseteq \subcells{\square}$, $T_i \subseteq \subcells{\square}$ and $S_i, T_i$ are $\eps$-well separated. Furthermore, for every pair of subcells $(\xi_1,\xi_2) \in \subcells{\square}\times \subcells{\square}$ (resp.$(\xi_2,\xi_1) \in \subcells{\square}\times \subcells{\square}$)   where $\xi_1$ and $\xi_2$ are in two different children of $\square$, there is a unique \emph{ordered} pair in $(X,Y) \in \wspd_{\square}$ (resp. $(Y,X) \in \wspd_{\square}$) such that $\xi_1 \in X$ and $\xi_2 \in Y$. We denote the ordered pair $(X,Y) \in \wspd_{\square}$ that the pair of sub-cells $(\xi_1,\xi_2)$ maps to as $(S_{\xi_1}, T_{\xi_2})$. For notational convenience, we prefer that the pairs within the WSPD are ordered.    Such an $\eps$-WSPD can be constructed by executing a standard quadtree based construction algorithm presented in~\cite{har2011geometric}. This algorithm uses the subtree of $Q$ rooted at $\square$ to build the WSPD. Since we are interested in $\xi_1$ and $\xi_2$ that are contained inside two different children of $\square$, we can trivially modify the algorithm of~\cite{har2011geometric} to guarantee that every pair $(S_i,T_i)$ in the WSPD is such that the subcells of $S_i$ and the subcells of $T_i$ are contained in two different children of $\square$. See Figure \ref{fig:subcells} for examples of WSPD pairs in $\wspd_\cell$. Finally, the algorithm of~\cite{har2011geometric} naturally generates unordered pairs. To ensure that every pair of subcells is covered by an ordered pair in the WSPD, for every pair $(X,Y) \in \wspd_{\square}$ generated by the algorithm, we add  $(Y,X)$ to $\wspd_{\square}$. 

Next, we define terms that will be helpful in describing our data structure in Section \ref{sec:ds}. Any point $p \in A\cup B$ is contained inside one cell of each of the grids $G_i$ in $Q$. Let $\cell=\cell_p^i$ be the cell of $G_i$ that contains $p$. Let  $\xi_{p}^\cell \in \subcells{\cell}$ be the subcell that contains $p$. As a property of the WSPD construction algorithm, the decomposition $\wspd_{\cell}$ ensures $\xi_p^{\cell}$ participates in $\BigOT(1)$ WSPD pairs of $\wspd_{\cell}$. Let this set be denoted by $N^i(p)$. All edges of level $i$ incident on $p$ are represented by exactly one pair in $N^i(p)$. Since there are $\BigO(\log n)$ levels, every edge incident on $p$ is represented by $\BigOT(1)$ WSPD pairs. We refer to these WSPD pairs as $N^*(p) = \bigcup_{i }N^i(p)$. 
 
 We can have a similar set of definitions for a subcell $\xi$ instead of a point $p$. Consider any cell $\cell \in G_i$ and a subcell $\xi \in \subcells{\cell}$. Using a similar argument, we conclude that all edges of level $i$ incident on any vertex of $(A\cup B)\cap \xi$ are uniquely represented by $\BigOT(1)$ WSPD pairs denoted by $N^i(\xi)$. Furthermore, all edges of level $\ge i$ are uniquely represented by $N^*(\xi) = \bigcup_{j\ge i} N^j(\xi)$. Note that $|N^*(\xi)|=\BigOT(1)$. 

\subsubsection{Distance Function} Given the definitions of subcells and the WSPDs, we can finally define the distance function. For $p,q \in A \cup B$, let $\square$ be the least common ancestor of $p$ and $q$ in $Q$ and let $i$ be the level of $\square$. We denote the \emph{level} of the edge $(p,q)$ to be the level of the least common ancestor of its end points, i.e., the level of $\square$. For some edge $(p,q)$ with least common ancestor $\cell$, let $\xi_p$ and $\xi_q$ be subcells from $\subcells{\square}$ that contain $p$ and $q$ respectively. Note that $\xi_p$ and $\xi_q$ are contained inside two different children of $\square$. There is a unique ordered \emph{representative pair} $(\Psi_p,\Psi_q)\in \mathcal{W}_{\cell}$ with $\xi_p \in \Psi_p$ and $\xi_q \in \Psi_q$.  We set the distance between $p$ and $q$ to be 
$$d_Q(p,q) = (\ell_{\max}(\Psi_p,\Psi_q))^2.$$
From the properties of our WSPD, if the unique representative pair of $(p,q)$ is $(X,Y)$, then the representative pair for $(q,p)$ will be $(Y,X)$, implying that our distance $d_Q(\cdot,\cdot)$ is symmetric.
For any subset $E\subseteq A\times B$ of edges, we define its cost by $d_Q(E) = \sum_{(a,b) \in E} d_Q(a,b)$. 
Since $p \in \xi_p, q \in \xi_q$ and $\xi_p\in \Psi_p$, $\xi_q \in \Psi_q$, we have
\begin{equation}
\label{eq:distlb}
    \distsq{p}{q} \le (\ell_{\max}(\xi_p,\xi_q))^2 \le (\ell_{\max}(\Psi_p,\Psi_q))^2 = d_Q(p,q).
\end{equation}
Furthermore, it can be shown that if $Q$ is a randomly shifted quad-tree, any optimal matching $\Mopt$ with respect to the original squared Euclidean costs satisfies
\begin{equation}
    \expect{d_Q(M_{\opt})} \leq (1+\varepsilon)\cdot\sum_{(a,b) \in \Mopt} \distsq{p}{q}. \label{eq:approxmatch}
\end{equation}
As noted before, we introduce an additional additive cost to all the edges. This additive cost on the edges is crucial in minimizing the number of data structure updates. Instead of proving~\eqref{eq:approxmatch}, in Section~\ref{sec:dualfeas}, we introduce this additional additive cost as part of the dual feasibility conditions.  We show that, to compute an $\eps$-approximate RMS matching, it suffices to compute a feasible perfect matching. 

\ignore{
\paragraph{Note on multisets:} Note that, for two points $a \in A$, $b\in B$ that share the same location, we have defined $d_Q(a,b)=0$. In order to simplify the notations in the rest of this paper, it is useful to extend some of the definitions above to points that have the same location. In particular, for any such $a, b$, let $\cell \in G_0$ be the node that contains $a$ and $b$. 
\begin{itemize}
    \item We define $\cell$ as their least common ancestor,
    \item Let $\xi \in \subcells{\cell}$ be the subcell that contains $a$ and $b$. We set $\xi_a = \xi_b = \xi$. 
    \item We set the minimum subcell size $\mu_0=0$ and also define $\ell_{\max}(\xi,\xi)=\ell_{\min}(\xi,\xi)$ to be  $0$.
    \item Finally, we create exactly one pair $\{(\xi,\xi)\}$ in $\wspd_{\cell}$.
\end{itemize}
It is straight-forward to verify that the useful properties of the WSPD established earlier in this section will also hold for $\wspd_{\cell}$.
}
\section{Dual Feasibility Conditions} 
\label{sec:dualfeas}
 In this section, we introduce a new set of feasibility conditions based on the randomly shifted quadtree. These feasibility conditions will allow our algorithm to find minimum-cost augmenting paths more efficiently. In order to describe this distance function, we partition the edges into a set of local edges and a set of non-local as described next. A similar definition of local and non-local edges was used in~\cite{sa_stoc12}.
 
\paragraph{Local and Non Local edges:} For any two matching edges $(a,b)\in M$ and $(a',b')\in M$, we say that they belong to the same equivalence class if and only if they have the same least common ancestor $\square$ and their ordered representative pairs in $\wspd_{\square}$ are the same, i.e., $(\Psi_a, \Psi_b)=(\Psi_{a'},\Psi_{b'})$.    Let $\K_M = \{M_1, ..., M_h\}$ be the resulting partition of matching edges into classes. For each $M_k$ for $1 \leq k \leq h$, let $A_k = \bigcup_{(a_j, b_j) \in M_k} a_j$ and $B_k = \bigcup_{(a_j, b_j) \in M_k} b_j$. The set $\{A_1,...,A_h\}$ partitions the matched vertices of $A$ and $\{B_1,...,B_h\}$ partitions the matched vertices of $B$.  For any edge $(a,b) \in A \times B$, we say $(a,b)$ is \textit{local} if $(a,b) \in A_k \times B_k$ for some $1 \leq k \leq h$. All other edges are \textit{non-local}. We refer to the local edges (both non-matching and matching) of $A_k \times B_k$ as \textit{class} $k$. 
\begin{figure}
    
    \centering
    \includegraphics[width=2\textwidth/3,height=\textheight,keepaspectratio]{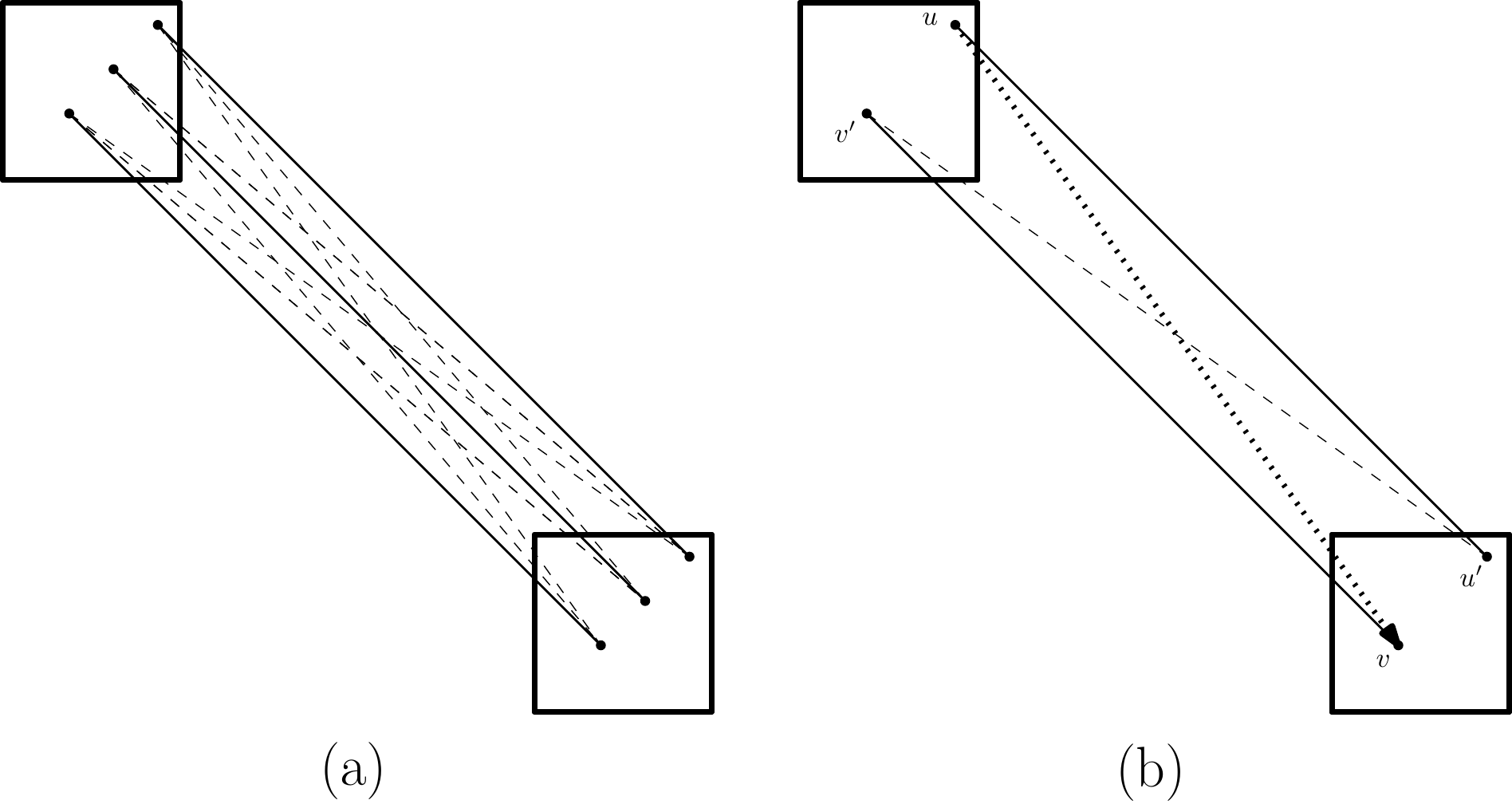}
    \caption{(a) A set of local edges between a WSPD pair of cells. Solid edges are in the matching, and dashed edges are not. (b) A local non-matching edge from $u \in B$ to $v \in A$ implies the existence of a length $3$ alternating path $P=\langle u,u',v',v\rangle$ with net-cost $\phi(P) = d_Q(u,v)$.}
    \label{fig:local-nonlocal}
\end{figure}

\ignore{For each $M_k$ for $1 \leq k \leq h$, let $\alpha_k = \bigcup_{(a_j, b_j) \in M_k} a_j$ and $\beta_k = \bigcup_{(a_j, b_j) \in M_k} b_j$. The set $\{\alpha_1,...,\alpha_h\}$ partitions the matched vertices of $A$ and $\{\beta_1,...,\beta_h\}$ partitions the matched vertices of $B$.  We  create a \emph{compact} graph with vertex sets $\mathbb{A}=\{\alpha_1,\ldots, \alpha_h\}\cup \bigcup_{a\in A_F} \{a\}$ and $\mathbb{B}=\{\beta_1,\ldots, \beta_h\}\cup \bigcup_{b \in B_F}\{b\}$. For any edge $(\alpha,\beta) \in \mathbb{A}\times\mathbb{B}$, we set its cost to be $d_Q(\alpha,\beta)=\min_{a \in \alpha, b\in \beta}d_Q(a,b)$. We then define a \emph{compact matching} $\mathbb{M}$ to be the edges $(\alpha_j,\beta_j)$ for all $1\le j\le h$ with a cost $d_Q(\alpha_j,\beta_j)$.  Note that any two pairs of points $(p, q), (p',q') \in \alpha_k\times \beta_k$ have the same cost, i.e., $d_Q(p,q)=d_Q(p',q')$ and so we can express the cost $d_Q(\alpha_j,\beta_j)$ by the cost of any pair $(p,q) \in \alpha_k\times\beta_k$, i.e., $d_Q(\alpha_k,\beta_k)=d_Q(p,q)$. We use $\alpha, \alpha'$ and $\alpha_j$ to denote points of $\mathbb{A}$ and $\beta, \beta'$ and $\beta_j$  to denote points of $\mathbb{B}$. We also extend the definition of level to any edge $(\alpha,\beta) \in \A\times \B$ to be the smallest value $i$ such that there is an edge $(a,b) \in \alpha \times \beta$ with level $i$. 
}
Next, we define a set of feasibility conditions based on the randomly-shifted quadtree. For a matching $M$ in the graph, $G(A\cup B, A\times B)$, we assign a \emph{dual weight} $y(v)$ for every $v \in A\cup B$. Recall that $\mu_i$ is the minimum subcell size at level $i$ in the quadtree. For any edge $(a, b)$ of level $i$, let $\mu_{ab} = \mu_{i}$. We say that a matching $M$ and set of dual weights $y(\cdot)$ are $Q$-feasible if for every edge $(a,b)$,
\begin{align}
    y(a) + y(b) &\leq d_Q(a, b) + \mu_{ab}^2. \label{eq:feas1}\\
    y(a) + y(b) &= d_Q(a, b) \text{\quad if } (a, b) \text{ is a local edge.} \label{eq:feas2}
\end{align}

A $Q$-feasible perfect matching is a \emph{$Q$-optimal matching}. Let $\Mopt$ be the optimal RMS matching in $G(A\cup B, A\times B)$. Similar to the Gabow-Tarjan~\cite{gt_sjc89} and Asathulla \etal~\cite{soda-18} algorithms, the addition of an additive error of $\mu_{ab}^2$ for non-local edges distorts the cost of non-local edges of $\Mopt$ by $\mu_{ab}^2$. However, it can be shown that this additional error for any non-local edge $(a,b)$ of the optimal matching is, in expectation, less than $\eps\distsq{a}{b}/2$ due to the random shift. This follows from the fact that short edges of the optimal matching have a small probability of appearing at higher levels of the quadtree. By combining this argument with properties of the distance function, we can show the following lemma, whose proof is delayed until Section~\ref{sec:ommittedproofs}: 

\begin{lemma}
\label{lem:distapprox}
For $A, B \subset \mathbb{R}^2$, let $\Mopt$ be the optimal RMS matching. For a parameter $\eps > 0$, given a randomly shifted quadtree $Q$ and the distance $d_Q(\cdot,\cdot)$, let $M$ be any $Q$-optimal matching. Then, 
$$  \expect{w(M)} \le (1+\eps/2)\sum_{(a,b) \in M_{\opt}}\distsq{a}{b}.$$
\end{lemma}
From Lemma \ref{lem:distapprox}, it follows that any $Q$-optimal matching is, in expectation, an $\eps$-approximate RMS matching. Therefore, it suffices to design an efficient algorithm for computing a $Q$-optimal matching. By executing such an algorithm $\BigO(\log{n})$ times, we can obtain an $\eps$-approximate RMS matching with high probability (see Section \ref{sec:ommittedproofs}).

\section{Algorithm}
\label{sec:algorithm}

\paragraph{Matching Preliminaries:} For any matching $M$, an \textit{alternating path} (resp. \textit{alternating cycle}) with respect to $M$ is one which alternates between edges of $M$ and edges not in $M$. A vertex is \textit{free} if it is not the endpoint of any edge of $M$ and \textit{matched} otherwise. We use $A_F$ (resp. $B_F$) to denote the set of free vertices of $A$ (resp. $B$). An \textit{augmenting path} $P$ is an alternating path between two free vertices. The matching $M' = M \oplus P$ has one higher cardinality than $M$.
An alternating path $P$  is called \textit{compact} if the largest contiguous set of local edges of $P$ has size at most $3$ (see Figure \ref{fig:local-nonlocal}).
Throughout this paper, we use the notation $a, a'$ and $a_j$ for $1\le j \le n$ to denote points in $A$ and $b, b'$ and $b_j$ for  $1\le j \le n$ to denote points in $B$.

For any non-local edge $(a, b)$, we define its \textit{slack} as $s(a, b) = d_Q(a, b) + \mu_{ab}^2- y(a) - y(b)$, i.e.,  how far the feasibility constraint~\eqref{eq:feas1} for $(u,v)$ is from holding with equality. For all local edges the slack $s(a,b)$ is defined to be $0$. Note that, for a $Q$-feasible matching, the slack on any edge is non-negative.  We say any edge is \textit{admissible} with respect to a set of dual weights if it has zero slack. The admissible graph is simply the subgraph induced by the set of zero slack edges. Note that all local edges are also admissible.

 As is common, we define the residual graph $\res{M}$ of a matching $M$ by assigning directions to edges of the graph $\mathcal{G}$. For any edge $(a,b) \in A \times B$, we direct $(a,b)$ from $a$ to $b$ if $(a,b) \in M$ and from $b$ to $a$ otherwise. 
 For any $Q$-feasible matching, we construct a weighted residual graph $\res{M}'$ where the edges of the graph are identical to $\res{M}$ and each edge $(a,b)$ has a weight equal to $s(a,b)$. Any path in $\res{M}$ is alternating, and any path in $\res{M}$ that starts with a free vertex of $B_F$ and ends at a free vertex of $A_F$ is an augmenting path. 
Our algorithm will maintain a $Q$-feasible matching $M$ and set of dual weights $y(\cdot)$. Initially $M = \emptyset$, and we set $y(v) \leftarrow 0$ for every vertex $v \in A \cup B$; clearly, this initial dual assignment is $Q$-feasible. Similar to the classical Hungarian algorithm, our algorithm will iteratively conduct a Hungarian search to find an augmenting path consisting only of admissible edges. Then the algorithm augments the matching along this path. This process repeats until a $Q$-optimal matching is found. Conducting a Hungarian search on the entire graph is prohibitively expensive. Therefore, we introduce a data structure that conducts Hungarian search and augment operations by implicitly modifying the dual weights in sub-linear time.

First, our algorithm executes in $\lceil 3\log n/4\rceil $ phases, starting with phase $0$. At the end of the execution of these phases, it produces a matching that has $\BigOT(n^{1/4})$ free vertices. Finally, the algorithm matches the remaining free vertices one at a time by conducting a Hungarian search to find an augmenting path and then augmenting the matching along this path.

At the start of any phase $i \geq 1$, we are given a $Q$-feasible matching $M$ along with a set of dual weights such that  every free vertex $b \in B_F$ has a dual weight of $\mu_{i-1}^2$.
At the end of phase $i$, we obtain a $Q$-feasible matching with the dual weights of any free vertex $b \in B_F$ risen to $\mu_{i}^2$.

The data structure is used only during the execution of phases. After the $\lceil 3\log n/4 \rceil$ phases have been executed, the algorithm will conduct explicit Hungarian searches and augmentations. For any phase $i \le \lceil 3\log n/4 \rceil$, we describe the data structure $\ds_i$. This data structure supports two global operations: 
\begin{itemize}
    \item \init\ : This operation takes as input a $Q$-feasible matching $M$ and a set of dual weights $y(\cdot)$ such that for every free vertex $v \in B_F$, the dual weight $y(v) = \mu_{i-1}^2$. Given $M, y(\cdot)$, the procedure builds the data structure. 
    \item \genduals\ : At any time in phase $i$, the execution of this procedure will return the matching $M$ stored by the data structure along with a set of dual weights $y(\cdot)$ such that $M, y(\cdot)$ is $Q$-feasible. We denote this matching as the \emph{associated $Q$-feasible matching}. 
\end{itemize}
 The total time taken by both of these operations is bounded by $\BigOT(n\mu_i^{2/3})$.

The data structure does not explicitly maintain a set of $Q$-feasible dual weights at all times because updating all the dual weights after each Hungarian search could take $\Omega(n)$ time. Instead the data structure maintains a smaller set of `up-to-date' dual weights, and updates other dual weights in a `lazy' fashion. While a similar strategy was used in \cite{soda-18}, applying the same strategy in our case requires the design of a new set of compressed feasibility conditions that are significantly more complex than the ones used in~\cite{soda-18}. An example of just one such complexity is the fact that our compressed feasibility relate vertices, edges, and dual weights defined across all levels of the quadtree, while the compressed feasibility conditions in~\cite{soda-18} do not require multiply `levels'.

A set of up-to-date $Q$-feasible dual weights for all vertices could be recovered after any Hungarian search or augmentation by simply executing \genduals. However, doing so is too expensive. Instead, the algorithm only executes \genduals\ once at the end of every phase. Nonetheless, the \genduals\ procedure guarantees the existence of a $Q$-feasible dual assignment for the matching $M$. We use this associated $Q$-feasible matching to describe the other operations supported by the data structure.

During phase $i$, we say that a cell $\cell \in G_i$ is \emph{active} if $\cell\cap B_F \neq \emptyset$. The edges that go between active cells have a cost of at least $\mu_i^2$ and do not become admissible during phase $i$ because the dual weights of all vertices of $B$ are at most $\mu_{i}^2$ whereas the points of $A$ have a non-positive dual weight. Therefore edges between active cells need not be considered during any Hungarian searches or augmentations of phase $i$. As a result, each Hungarian search and augmentation can be conducted completely within a single active cell. 

During any phase $i$ and for any active cell $\square^* \in G_i$, our data structure supports the following operations:
\begin{itemize}
    \item \search\ : This procedure conducts a Hungarian search. At the end of the search, either the dual weight of every free vertex $b \in B_F\cap \square^*$ with respect to the associated $Q$-feasible matching has risen to $\mu_{i}^2$, or the search returns an augmenting path $P$ inside $\square^*$ such that $P$ is both admissible and compact with respect to the associated $Q$-feasible matching. 
    \item \augment\ : This procedure augments the matching along an augmenting path $P$ returned by \search\ and updates the data structure to reflect the new matching. 
\end{itemize}
We postpone the implementation details of the four operations supported by this data structure until Section \ref{sec:ds}. 

In order to describe the execution time of these procedures, we define  \emph{active trees} next.

\begin{figure}
    
    \centering
    \includegraphics[width=\textwidth,height=\textheight,keepaspectratio]{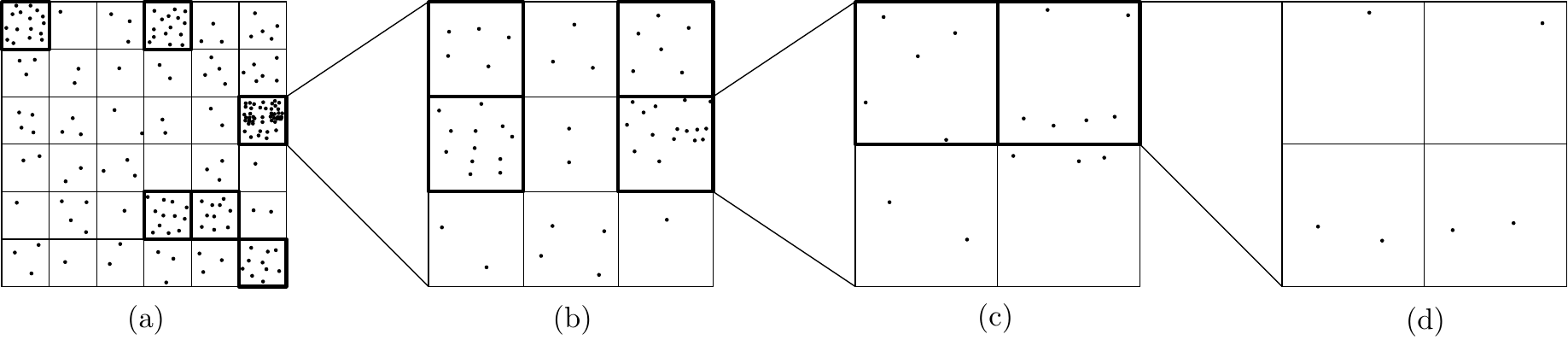}
    \caption{(a) The active cells of the current layer. Each full active cell (bold) has descendants in its active tree. (b) The pieces of a full active cell. (c) Each full piece has $4$ descendants in the active tree. (d) Each branch of the active tree terminates with a sparse leaf cell.}
    \label{fig:activetree}
\end{figure}
\paragraph{Active Tree:} We say a cell $\cell$ at level $i$ in the quadtree $Q$ is \emph{sparse} if  $|(A\cup B)\cap \cell| \leq \mu_i^2$ or if $\cell \in G_0$. Otherwise, $\cell$ is \emph{full}.
For each active cell $\cell^*$ during phase $i$, we maintain an \emph{active tree} denoted by $\mathcal{T}_{\cell^*}$. The active tree $\mathcal{T}_{\cell^*}$ is rooted at $\cell^*$ and contains a subset of the nodes in the subtree of $\cell^*$ in $Q$. If $\cell^*$ is sparse, then $\cell^*$ is also a leaf node and the active tree contains only one node. Otherwise, if $\cell^*$ is full, let all cells of $G_{\lfloor 2i/3\rfloor}$ that partition $\cell^*$ be the children of $\cell^*$ in the active tree. We refer to every child of $\cell^*$ in the active tree as a \emph{piece} of $\cell^*$. 
For each piece $\cell$ of $\cell^*$, if $\cell$ is sparse, then $\cell$ will become a leaf node of the active tree. Otherwise, if $\cell$ is full, then $\cell$ is an internal node of the active tree, and the four children of $\cell$ in $Q$ are also contained in the active tree. We recursively apply this process to construct the active tree for each of the four children; each full child is decomposed into its four children in $Q$. Every leaf node of $\mathcal{T}_{\cell^*}$ is a sparse cell and every internal node is a full cell.

 Consider any augmenting path $P$ computed inside an active cell $\square^*$ during phase $i$. Let $\aff{P}$ be the set of all cells of the active tree, excluding $\square^*$, that contain at least one vertex of $P$. We call such cells the \emph{affected} cells of $P$. Let $\affj{P}{j}$ be the set of level $j$ affected cells of $P$. Then the time taken for a single execution of the \search\ procedure that returns an augmenting path $P$ is $\BigOT(\mu_{i}^{8/3} +\sum_{j=0}^{\lfloor 2i/3\rfloor}|\affj{P}{j}|\mu_j^3)$ (see Section~\ref{subsec:hungariansearch}), and the time taken for an execution of the \augment\ procedure on an augmenting path $P$ is $\BigOT(\sum_{j=0}^{\lfloor 2i/3\rfloor}|\affj{P}{j}|\mu_j^3)$ (see Section~\ref{subsec:augment}). 

Using these operations, we now present our algorithm for any phase $0\le i < \lceil 3\log n/4\rceil $. At the start of phase $i \geq 1$, the dual weight of every vertex in $B_F$ is equal to $\mu_{i-1}^2$. We mark all active cells as unprocessed. The algorithm for phase $i$ conducts the following steps.
 \begin{itemize}
     \item Build the data structure $\ds$ using the \init\ procedure.
     \item While there is an unprocessed active cell $\cell^*$,
     \begin{itemize}
        \item Execute the \search\ procedure on $\cell^*$. 
        \item If \search\ returned an augmenting path $P$, then execute \augment\ on $P$.
        \item If either $B_F\cap \square^*=\emptyset$ (i.e., $\square^*$ is no longer active) or the dual weight of every vertex of $B_F\cap \square^*$ is $\mu_{i}^2$, then mark $\square^*$ as processed.
     \end{itemize}
     \item Use \genduals\ to obtain the associated $Q$-feasible matching $M$ and the dual weights $y(\cdot)$.  
 \end{itemize}
 After the execution of all $\lceil 3\log n/4 \rceil$ phases, we match the remaining free vertices one at a time by iteratively executing Hungarian search to find an augmenting path and augmenting the matching $M$ along the path. Unlike during the phases, each Hungarian search is global and executed on $\res{M}$ without use of the data structure. We describe the details of this global Hungarian search next.

\paragraph{Hungarian search:}
We add a source vertex $s$ to the graph $\res{M}'$ and connect $s$ to each vertex $b \in B_F$ with a cost $0$ edge. Then, we execute a Dijkstra search in the resulting graph, starting from $s$. For any point $u \in A \cup B$, let $d_u$ be the shortest path distance from $s$ to $u$ as computed by Dijkstra's algorithm. We define as value $d$ as,
\[d = \min_{f \in A_F} d_f.\]
Next, for every $u \in A \cup B$ with $d_u\leq d$, we perform the following dual adjustment. If $u \in B$, we set $y(u) \leftarrow y(u) + d - d_u$, and if $u \in A$, we set $y(u) \leftarrow y(u) - d + d_u$. Using a straightforward and standard argument, it is easy to show that this dual adjustment maintains $Q$-feasibility. Furthermore, after the Hungarian search, $\res{M}$ contains an augmenting path $P$ consisting solely of admissible edges. 

When implemented naively, the Hungarian search could take $\Omega(n^2)$ time. However, we recall that each vertex of $\res{M}$ is part of only $\BigOT(1)$ WSPD pairs. All edges $(u,v)$ with the same representative WSPD pair have the same value of $d_Q(u,v)$ as well as the same value of $\mu_{uv}^2$. Using this fact, Dijkstra's algorithm can efficiently find the next edge to add to the shortest path tree in amortized $\BigOT(1)$ time per addition. As a result, a single Hungarian search can be executed in $\BigOT(n)$ time.

\paragraph{Augment:} Let $P$ be an admissible augmenting path found by the Hungarian search procedure. We describe how to augment $M$ along $P$ while maintaining $Q$-feasibility. First, we set $M \leftarrow M \oplus P$. This causes some non-local edges (potentially both matching and non-matching) to become local. We must adjust the dual weights along $P$ to ensure that every newly introduced local edge $(u,v)$ satisfies the $Q$-feasibility constraint $y(u) + y(v) = d_Q(u,v)$. Let $(a,b) \in A_k \times B_k$ be an edge of $P$ that is local and in class $k$ edge after augmentation, but was non-local prior to augmentation. If there are no other class $k$ local edges after augmentation that were also local prior to augmentation, we set $y(a) \leftarrow y(a) - \mu_{ab}^2$. Otherwise, there must be at least one local edge $(a',b')$ in class $k$ after augmentation that was local prior to augmentation, and we set $y(a) \leftarrow y(a')$ and $y(b) \leftarrow y(b')$.

\paragraph{Invariants:}
During the execution of the phases, the algorithm guarantees the following invariants:

\begin{itemize}
    \item[(I1)] The associated matching $M, y(\cdot)$ is $Q$-feasible.
    \item[(I2)] The dual weight $y(b)$ of every vertex $b \in B$ is non-negative. Furthermore, in phase $i \geq 1$, the dual weight of every free vertex $b \in B_F$ is $\mu_{i-1}^2 \le y(b) \le \mu_{i}^2$. The dual weight of every vertex $a \in A$ is non-positive and for every free vertex $a \in A_F$, $y(a)=0$.
\end{itemize}
Since each step of the algorithm is a call to the data structure $\ds_i$, it suffices to show that $\ds_i$ maintains these invariants. We do this in Section~\ref{sec:ds}.

After the execution of the phases, the algorithm switches to conducting explicit Hungarian searches. Using a standard argument, it is easy to show that the dual updates of these Hungarian searches maintain $Q$-feasibility. However, the dual adjustments during augmentations are non-standard due to a careful handling of local edges. The following lemmas establish that the augmentation process continues to maintain $Q$-feasibility. Therefore, at the end of the algorithm, we produce a $Q$-optimal matching as desired.

\begin{lemma}
\label{lem:localsamedual}
Any pair of local edges $(u,v) \in A \times B$ and $(u',v') \in A \times B$ of the same class $k$ are $Q$-feasible if and only if $y(u) = y(u')$ and $y(v) = y(v')$.
\end{lemma}
\begin{proof}
It is sufficient to argue that claim is true for a pair of matching class $k$ edges, since any local non-matching edge has both its endpoints matched by a class $k$ matching edge. If $(u,v)$ and $(u',v')$ are matching edges of class $k$, then there must also be a pair of non-matching local edges $(v,u')$ and $(v',u)$. Since $(u,v)$ and $(v,u')$ are both feasible, we have
\[y(u) = d_Q(u,v) - y(v) = d_Q(v, u') - y(v) = y(u').\] Similarly, since $(u',v')$ and $(v,u')$ are both feasible, we have. 
\[y(v') = d_Q(u',v') - y(u') = d_Q(v, u') - y(u') = y(v).\]
\end{proof}

\begin{lemma}
\label{lem:augmentdecrease}
For any vertex $v \in A \cup B$, let $y(v)$ be the dual weight of $v$ immediately prior to  augmenting along a path $P$. Along with this augmentation, the dual weights of some vertices on $P$ are modified. Let $y'(v)$ be the new dual weights after these modifications. Then, $y'(v) \leq y(v)$. 
\end{lemma}
\begin{proof}
The only dual weights that change are along $P$. Any vertex $u$ on $P$ must be matched to a vertex $v$ that is also on $P$ after augmentation. If $(u,v)$ is a local edge that was non-local prior to augmentation, and there are no edges that were in class $k$ both before and after augmentation, then the procedure sets $y'(u) = y(u) - \mu_{uv}^2$ if $u \in A$, and if $u \in B$, then its dual weight is left unchanged. It is easy to see that the dual weights of $u$ and $v$ do not increase for this case. Otherwise, there must be some other edge $(u',v')$ that is local and in class $k$ both before and after augmentation. Since $(u,v)$ was an admissible non-local edge prior to augmentation, 
\[y(u) = d_Q(u,v) + \mu_{uv}^2 - y(v) \geq d_Q(v, u') - y(v),\]
and the dual weight of $u$ only decreases when it is set to match $y(u')$. 
\end{proof}
\begin{lemma}
\label{lem:augment}
Let $P$ be an admissible path with respect to a $Q$-feasible matching $M$ and set of dual weights $y(\cdot)$. Let $M', y'(\cdot)$ be the matching and set of dual weights after augmenting along $P$. Then $M', y'(\cdot)$ are $Q$-feasible.
\end{lemma}
\begin{proof}
Since the augmentation process only changes the dual weights of vertices of $P$, we only need to consider edges that have at least one endpoint on $P$; edges disjoint from $P$ are unaffected. First, consider any local edge $(u,v)$ of class $k$ after augmentation that was non-local prior to augmentation. If, after augmentation, there are no other class $k$ local edges that also existed prior to augmentation, then either (i) $(u,v)$ is on $P$, or (ii) $(u,v)$ is not on $P$. For case (i), the edge $(u,v)$ is on $P$, $(u,v)$ is a matching edge after augmentation, $u \in A$, $v \in B$, and the procedure sets $y'(u) = y(u) - \mu_{uv}^2$. Prior to augmentation, there was an admissible non-local edge directed from $v$ to $u$, and we have $y(u) + y(v) = d_Q(u,v) + \mu_{uv}^2$. Therefore, after augmentation, $y'(u) + y'(v) = d_Q(u,v)$, and $(u,v)$ is feasible. Case (ii) can only occur if two previously non-local edges $(u, v')$ and $(u',v)$ simultaneously enter class $k$ by augmenting along $P$; in this case, it must be true that $y(u') = y(u)$ and $y(v')=y(v)$ and the edge $(u,v)$ is feasible.

Next, consider the case where there is at least one other class $k$ local edge after augmentation. At least one such edge  $(u',v')$ must have been in class $k$ prior to augmentation as well. Then, for any newly created class $k$ local edge, the procedure ensures that $y'(u) = y(u')$ and $y'(v) = y(v')$. This implies that all class $k$ local edges are feasible by Lemma \ref{lem:localsamedual}.

Finally, we argue that any non-local edge $(u,v)$ after augmentation is feasible. From Lemma \ref{lem:augmentdecrease}, we have 
\[y'(u) + y'(v) \leq y(u) + y(v) \leq d_Q(u,v) + \mu_{uv}^2.\]
\end{proof}
\subsection{Analysis of the algorithm}
\label{subsec:analysisofalg}
In this section, we bound the time taken by the algorithm under the assumption that the data structure works as described. We begin by defining notations that will be used throughout the analysis. Let $\mathbb{P}=\langle P_1,\ldots, P_t\rangle$ be the $t$ augmenting paths computed during the $\lceil 3\log n/4\rceil$ phases of the algorithm. Let $M_0$ be the initial empty matching and, for any $k \ge 1$, let $M_k$ be the matching obtained after augmenting the matching $M_{k-1}$ along $P_k$, i.e., $M_{k}=M_{k-1}\oplus P_k$. For any augmenting path $P$ with respect to some matching $M$, let $N(P)$ be the set of non-local edges of $P$. 

The following Lemma establishes important properties of the algorithm during the  $\lceil 3 \log{n}/4 \rceil$ phases of the algorithm.
\begin{lemma}
\label{lem:unmatchedrem} 
\label{lem:totalerror}
The algorithm maintains the following properties during the $\lceil 3 \log{n}/4 \rceil$ phases:
\begin{enumerate}[(i)]
    \item The total number of free vertices remaining at the end of phase $i$ is $\BigOT(n/\mu_i^2)$, and,
    \item   $\sum_{k=1}^t\sum_{(a,b) \in N(P_k)} \mu_{ab}^2 = \BigOT(n)$.
\end{enumerate}
\end{lemma}
\begin{proof}
First, we consider phase $0$. Clearly (i) holds because the number of unmatched vertices at the end of phase $0$ is $\BigO(n)$. It is also easy to show that every augmenting path found during phase $0$ contributes only $\BigO(1)$ to the total given by property (ii), and phase $0$ contributes only $\BigO(n)$ towards (ii) over all its augmenting paths. Therefore, in the remaining arguments, we will assume that $i \geq 1$.

Let $M_{\opt}$ be the minimum cost matching. The symmetric difference of $M$ and $M_{\opt}$ will contain $\ell=n-|M|$ vertex-disjoint augmenting paths. Let $\{\mathcal{P}_1,\ldots, \mathcal{P}_\ell\}$ be these augmenting paths. These augmenting paths contain some of the edges of $M_{\opt}$. Combining this with the $Q$-feasibility conditions gives,
\begin{eqnarray*}
\sum_{(a,b) \in M_{\opt}} (d_Q(a,b)+\mu_{ab}^2) &\ge & \sum_{k=1}^\ell\left( \sum_{(a,b) \in \mathcal{P}_i \setminus M} (d_Q(a,b)+\mu_{ab}^2)  - \sum_{(a,b) \in \mathcal{P}_i \cap M} d_Q(a,b) \right)\\
&\ge& \sum_{k=1}^\ell\left( \sum_{(a,b) \in \mathcal{P}_i \setminus M} (y(a)+y(b))  - \sum_{(a,b) \in \mathcal{P}_i \cap M} (y(a)+y(b)) \right)\\
&\geq& \ell \mu_{i-1}^2.
\end{eqnarray*}
The last inequality follows from the facts that, at the end of phase $i-1$, the dual weight of every vertex in $B_F$ is at least $\mu_{i-1}^2$, the dual weight of every vertex of $A_F$ is $0$, and all the vertices of any augmenting path except the first and last are the endpoint of exactly one matching edge and exactly one non-matching edge. At the beginning of phase $i$, since $\sum_{(a,b) \in M_{\mathrm{OPT}}} (d_Q(a,b)+\mu_{ab}^2)$ is $\BigOT(n)$, the number of free vertices at the end of phase $i$ is $\BigOT(n/\mu_{i-1}^2)$. Noting that $\mu_{i-1} \leq \mu_{i}$ gives (i).

Next, we prove (ii). Recollect that $\{P_1,\ldots, P_t\}$ are the augmenting paths computed by the algorithm. For any path $P_k$, let $b_k$ and $a_k$ be the two endpoints of this augmenting path and $y_k$ be the dual weight of $b_k$ when the augmenting path $P_k$ was found. Suppose $P_k$ was found in some phase $i$. Then, its dual weight $y_k \le \mu_i^2 \le 4\mu_{i-1}^2 \le \BigOT(n)/(n-k)$. Summing over $1\le k \le t$, we get
$$\sum_{k=1}^t y_k = \BigOT(n).$$
Note that because $d_Q(M_0) = 0$, we have,
\begin{eqnarray*}
d_Q(M_t) +\sum_{k=1}^t\sum_{(a,b) \in P_k\setminus M_{k-1}} \mu_{ab}^2 &\le& \sum_{k=1}^t d_Q(M_k) - d_Q(M_{k-1}) + \sum_{(a,b) \in P_k\setminus M_{k-1}} \mu_{ab}^2 \\
&\le& \sum_{k=1}^t \left(\sum_{(a,b) \in P_k \setminus M_{k-1}} (d_Q(a,b) +\mu_{ab}^2) - \sum_{(a,b) \in P_k \cap M_{k-1}} d_Q(a,b)\right)\\
&=& \sum_{k=1}^t \left(\sum_{(a,b) \in P_k \setminus M_{k-1}}y(a)+y(b) - \sum_{(a,b) \in P_k \cap M_{k-1}} y(a)+y(b)\right)\\
&=& \sum_{k=1}^t y_k = \BigOT(n).
\end{eqnarray*} 
The claim then follows from the facts that $d_Q(M_t) \geq 0$ and $N(P_k) \subseteq (P_k \setminus M_{k-1})$.
\end{proof}

Using Lemma \ref{lem:unmatchedrem}, we can bound the efficiency of the algorithm. First, we bound the total time taken after the $\lceil 3\log n/4\rceil$ phases have been executed. After the last phase is executed, $\mu_{i}^2 = \Omega(n^{3/4} / \polys)$. From Lemma \ref{lem:unmatchedrem}, there are only $\BigOT(n^{1/4})$ unmatched vertices remaining. Using the WSPD, each of these unmatched vertices are matched in $\BigOT(n)$ time. Therefore, the time taken after the phases have executed is $\BigOT(n^{5/4})$.

Next, we bound the time taken by the $\lceil 3\log n/4\rceil$ phases of the algorithm. For any such phase $i$, we execute the \init\ procedure to create the data structure $\ds_i$. At the end of phase $i$, we execute the \genduals\ procedure to generate a $Q$-feasible matching. Both of these operations take $\BigOT(n\mu_i^{2/3})$ time during phase $i$.

Next, we bound the time taken by \augment, which takes $\BigOT(\sum_{j=0}^{\lfloor 2i/3\rfloor}|\affj{P}{}|\mu_j^3)$ time when executed on an augmenting path $P$. Therefore, to bound the total time taken by \augment, we will bound the total  number of level $j$ edges over all augmenting paths computed during phases of the algorithm. From Lemma \ref{lem:totalerror}, we have
\begin{equation}
\label{eq:totalerror2}
\sum_{1 \leq k \leq t}\sum_{(u,v) \in N(P_k)} \mu_{uv}^2 = \BigOT(n).
\end{equation}
Each non-local edge $(u,v)$ of level $j$ contributes $\Omega(\mu_{uv}^2) = \Omega(\mu_{j}^2)$ towards the RHS of equation \eqref{eq:totalerror2}. As a result, there can be at most $\BigOT(n/\mu_j^2)$ such edges in all augmenting paths computed during the phases of the algorithm. 

Recall that two matching edges $(a_i,b_i)$ and $(a_k,b_k)$ are in the same class if they share the least common ancestor $\square$ and their representative pair $(\Psi_{a_i},\Psi_{b_i}) \in \wspd_{\square}$ is the same as $(\Psi_{a_k}, \Psi_{b_k})$. Consider any augmenting path $P$, and consider any maximal sub-path $S$ with the property that all its matching edges (resp. non-matching edges) belong to the same class with representative pair $(\Psi,\Psi')$ (resp. $(\Psi',\Psi)$). We will call any such path a \emph{local path}. Intuitively, all matching edges of $S$ will belong to the same class and, upon augmentation, the non-matching edges of $S$ will all enter the matching and belong to the same class.  We say that $S$ is a level $j$ local path if all edges of $S$ appear at level $j$. For any augmenting path $P$, let $L_j(P)$ be the set of level $j$ local paths of $P$.
\begin{figure}
    \centering
    \includegraphics[width=2\textwidth/3,height=\textheight,keepaspectratio]{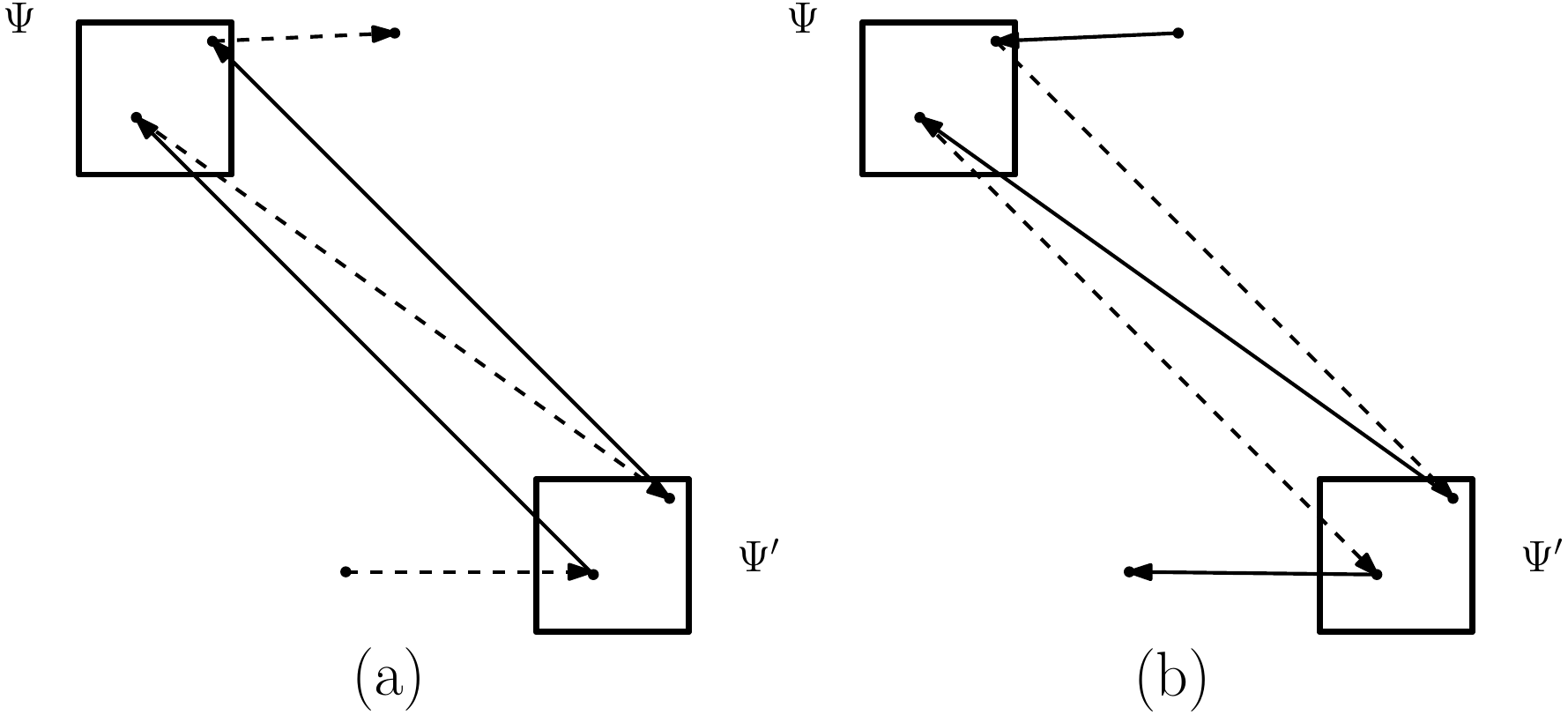}
    \caption{(a) A local path $S$ of length $3$ between the cells $\Psi$ and $\Psi'$. Although $S$ does not contain any non-local edges, augmentation will reduce the number of matching edges between the two cells by $1$. (b) The number of matching edges between the cells $\Psi$ and $\Psi'$ can only increase when a local path containing at least one non-local edge between $\Psi$ and $\Psi'$ participates in an augmenting path.}
    \label{fig:local-path}
\end{figure}
 It is easy to see that either: (1) $S$ contains at least one non-local edge or (2) the first and the last edge of $S$ are matching edges. In case (2), the number of matching edges that have $(\Psi, \Psi')$ as their representative pair decreases by $1$ after augmenting along $P$. Furthermore, new matching edges with representative $(\Psi, \Psi')$ can only be created through an occurrence of case (1). Therefore, each occurrence of case (2) can be taxed on an occurrence of case (1). Combining this observation with the bound on the number of non-local edges gives:
\begin{equation}
\label{eq:localgroupbound}
\sum_{1 \leq k \leq t}|L_j(P_k)| = \BigOT(n / \mu_j^2).
\end{equation}
From the fact that each $P_k$ is compact, there are at most $6|L_j(P_k)|+1$ cells in $\affj{P_k}{j}$. From Lemma \ref{lem:unmatchedrem}, there are $\BigOT(n/\mu_i^2)$ augmenting paths found during phase $i$. Combining these observations with \eqref{eq:localgroupbound} gives the following for any level $j$,
\[\sum_{k=1}^t|\affj{P_k}{j}|\mu_j^3 = \BigOT(\mu_j^3(n/\mu_i^2 + n/\mu_j^2))  = \BigOT(n\mu_j).\]
When summing over all levels that contain affected pieces, the top level $\lfloor 2i/3 \rfloor$ dominates, which bounds the total time for \augment\ during phase $i$ as,
\begin{equation}
\label{eq:totalupdate}
    \sum_{k=1}^t\sum_{j=0}^{\lfloor 2i/3\rfloor} |\affj{P_k}{j}|\mu_j^3 = \BigOT(n\mu_i^{2/3}).
\end{equation}
\ignore {
\paragraph{Note concerning edges of layer $0$:} Any non-local edge $(a,b)$ of layer $0$ has a value $\mu_{ab} = 0$. As a result, a slightly different argument is needed to bound the total number of layer $0$ edges used on augmenting paths. Observe that the edges preceding or following any local path of level $0$ must be part of some level greater than $0$. As a result, the total number of layer $0$ edges used across all augmenting paths can be taxed on the total number of non-layer $0$ edges across all augmenting paths, which is $\tilde{O}(n)$.
}
Finally, we bound the time taken by the \search\ procedure during phase $i$. 
From Lemma \ref{lem:unmatchedrem}, the number of unmatched vertices remaining at the beginning of phase $i$ is $\BigO(n/\mu_i^2)$. This value also bounds the number of active cells during phase $i$. Therefore, the number of calls to \search\ during phase $i$ is $\BigO(n/\mu_i^2)$. Each execution of \search\ during phase $i$ takes $\BigOT(\mu_{i}^{8/3} +\sum_{j=0}^{\lfloor 2i/3\rfloor}|\affj{P}{j}|\mu_j^3)$ time. Summing over all paths computed during phase $i$ and applying \eqref{eq:totalupdate} gives a total time of $\BigOT(n\mu_i^{2/3})$ for the \search\ procedure during phase $i$. Combining the times taken by all data structure procedures during phase $i$ gives a total time of $\BigOT(n\mu_i^{2/3})$. The time taken for the last phase $\lceil 3\log{n}/4 - 1\rceil$ dominates, taking a total of $\BigOT(n^{5/4})$ time.

 \section{Data Structure}
 \label{sec:ds}
 \subsection{Preliminaries}
 To simplify the presentation of the data structure, we introduce additional notations and give an equivalent redefinition of $Q$-feasibility with respect to these additional notations. We also present a few auxiliary properties that will be useful in proving the correctness of the data structure.
 
 We define the \textit{adjusted} cost of an edge $\Phi(a,b)$ as 
\begin{align*}
    &\Phi(a,b) = d_Q(a,b) + \mu_{ab}^2&\text{ if } (a,b) \text{ is non-local.}\\
    &\Phi(a,b) = d_Q(a,b) &\text{ otherwise. }
\end{align*}
For any edge $(a,b)$ of $\mathcal{G}$, we define its \textit{net-cost} $\phi(a,b)$ as follows. If $(a,b)$ is non-matching edge, its net-cost is $\phi(a,b) = \Phi(a,b)$. Otherwise, $(a,b)$ if the edge is in the matching, we define $\phi(a,b) = -\Phi(a,b)$. For any set of edges $S$, we define its \textit{net-cost} as $\phi(S) = \sum_{(a,b) \in M} \phi(a,b)$. 

Recollect that $Q$-feasibility was defined with respect to the graph $\mathcal{G}(A\cup B, A \times B)$.  For the data structure, it is convenient to deal with the residual graph $\res{M}$ instead. We redefine a $Q$-feasibility constraints that our algorithm maintains.  This is done for simplicity in exposition of the algorithm and its proofs. All dual updates done during the course of the algorithm will ensure that each point $b \in  B$ is assigned a non-negative dual weight $y(b)$ and each point $a \in A$ is assigned a non-positive dual weight $y(a)$.

A matching $M$ and a set of dual assignments $y(\cdot)$ is $Q$-feasible if for any edge $(u,v)$ of the residual graph $\res{M}$ directed from $u$ to $v$,  
\begin{align*}
    |y(u)| - |y(v)| &\leq \phi(u,v),\\
    |y(u)| - |y(v)| &= \phi(u,v) &\text{ if } (u,v) \text{ is local.}
\end{align*}For any non-local edge $(u,v)$, we define its \textit{slack} as $s(u,v) = \phi(u,v) - |y(u)| + |y(v)|$, i.e.,  how far the feasibility constraint for $(u,v)$ is from being violated. Note that the slack on any edge is non-negative with local edges having a zero slack. We say any edge is \textit{admissible} with respect to a set of dual weights it has a zero slack. The admissible graph is simply the subgraph induced by the set of zero slack edges. The advantage of redefining $Q$-feasibility conditions in this fashion is that it extends to any directed path in $\res{M}$ as presented in the following lemma. 

\begin{lemma}
\label{lem:dualsslacknetcost}
Let $M, y(\cdot)$ be a $Q$-feasible matching and set of dual weights maintained by the algorithm. Let $P$ be any alternating path with respect to $M$ starting at a vertex $u$ and ending at a vertex $v$. Then,
\[|y(u)| - |y(v)| + \sum_{(u',v') \in P} s(u',v') = \phi(P).\]
\end{lemma}
\begin{proof}
The proof is straight-forward from the definitions of slack and net-cost:
\[\phi(P) = \sum_{(u',v') \in P}\phi(u',v') = \sum_{(u',v') \in P} (|y(u')| - |y(v')| + s(u',v') ) = |y(u)| - |y(v)| + \sum_{(u',v') \in P} s(u',v').\]
\end{proof}

For any phase $i$, we define our data structure for each active cell $\cell^*\in G_i$. The data structure is based on the active tree $\mathcal{T}_{\cell^*}$.  We define a sub-linear in $n$ sized associated graph $\mathcal{AG}_{\cell}$ for each cell $\cell$ of $\mathcal{T}_{\cell^*}$. This graph will help us compactly store the dual weights and help conduct the Hungarian Search, find augmenting paths and augment the matching along the path.

For any cell $\cell$, let $A_\cell = A \cap \cell$ and $B_\cell = B \cap \cell$. For any set of cells $X$, we denote $A_X = \bigcup_{\cell \in X}A_{\cell}$ and $B_X= \bigcup_{\cell \in X}B_{\cell}$. For any cell $\cell$ let $M_\cell$ be the set of edges of $M$ that have both endpoints contained in $\cell$, and let $\res{M_\cell}$ be the vertex-induced subgraph of $(A_\cell \cup B_\cell)$ on $\res{M}$. For simplicity in notation we use $\res{\cell}$ to denote $\res{M_\cell}$. 

Recall that a $\cell$ of level $j$ is \emph{sparse} if $|A_\cell \cup B_\cell| \leq \mu_j^2$ or $\cell \in G_0$ and full otherwise. For any cell $\cell$ of $\mathcal{T}_{\cell^*}$, our data structure constructs an \emph{associated graph} $\mathcal{AG}_{\cell}$. If $\cell$ is sparse, the associated graph $\mathcal{AG}_{\cell}$ is simply given by $\res{\cell}$. In the following, we define the associated graph for any full cell.

\subsection{Vertices of the associated graph}
\label{subsec:associatedgraphvertices}
 We define an associated graph for an arbitrary cell in the active tree $\mathcal{T}_{\cell^*}$. For any cell $\cell$ in the active tree, let $D(\cell)$ denote the children of $\cell$. We extend our definition to subcells as well. For any subcell $\xi \in \subcells{\cell}$, let $\cell'\in D(\cell)$ be the cell that contains $\xi$. Let $D(\xi) = \{\xi' \mid \xi' \in \subcells{\cell'} \text{ and }\xi' \subseteq \xi\}$. 
 
 If $\cell$ is a full cell, for each of its children $\cell'$, suppose $\cell'$ is of level $j$. We cluster $A_{\cell'}\cup B_{\cell'}$ into $\BigOT(\mu_{j})$ clusters. We cluster points in such a way that all edges going between any two clusters $X$ and $Y$, where $X$ and $Y$ are clusters for two different children $\cell'$ and $\cell''$ of $\cell$, have the same net-cost. We create one vertex in $V_{\cell}$ for every cluster of $\cell'$ and repeat this for every child $\cell'$ of $\cell$.  The clusters created here are similar to that in \cite{sa_stoc12}. 

Recall that two matching edges $(a_i,b_i)$ and $(a_k,b_k)$ are in the same class if they share the least common ancestor $\square$ and their representative pair $(\Psi_{a_i},\Psi_{b_i}) \in \wspd_{\square}$ is the same as $(\Psi_{a_k}, \Psi_{b_k})$. 
For any matched point $a_i$ (resp.\ $b_i)$, we refer to $b_i$ (resp.\ $a_i$) as
its \emph{partner} point. 
For any $\subcell \in \Gd[\cell']$, we partition $A_{\subcell}$ and $B_{\subcell}$ into three types of clusters.
\begin{figure}
    \centering
    \includegraphics[width=\textwidth]{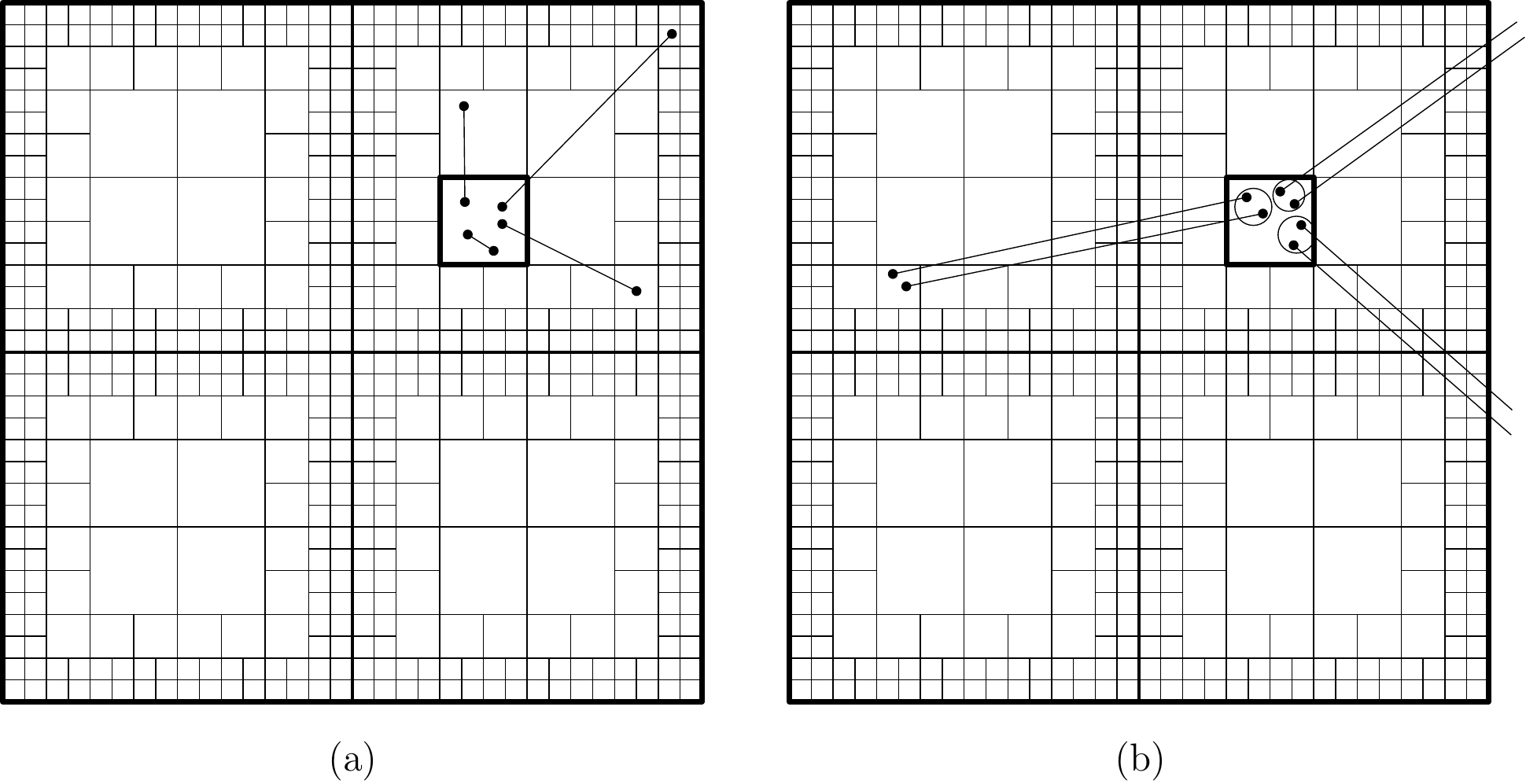}
    \caption{Internal and boundary clusters of a subcell $\xi$ within the child $\cell'$ of $\cell$. (a) All vertices matched within $\cell'$ are part of an internal cluster. (b) All vertices matched outside of $\cell'$ are divided into boundary clusters based on WSPD pairs at higher levels.  }
    \label{fig:clusters}
\end{figure}
\begin{itemize}
\item {\it Free clusters}: All free points of $A_\subcell$ (resp. $B_\subcell$) belong to a single cluster
$$A_\subcell^F = A_F\cap \subcell,\quad B_\subcell^F = B_F \cap \subcell.$$
\item {\it Internal clusters}: All points of $A_\subcell$ (resp. $B_\subcell$) whose partner point is also inside $\cell'$ belong to a single cluster
$$A_\subcell^I = \{a_i \in A_\subcell \mid (a_i,b_i) \in M, b_i \in B_{\cell'}\},$$ 
$$B_\subcell^I = \{b_i \in B_\subcell \mid (a_i,b_i) \in M, a_i \in A_{\cell'}\}.$$ 
\item {\it Boundary clusters}: Recollect that $\cell' \in G_j$. All points of $A_\subcell$ (resp. $B_\subcell)$ whose partner points are outside $\cell'$ are partitioned into \emph{boundary clusters}. Two such vertices belong to the same boundary cluster if the matching edges incident on them belong to the same class. Note that all such matching edges have level at least $j$ and are incident on at least one vertex of  $A_\xi\cup B_\xi$. Any such matching edges is captured by one of $\BigOT(1)$ many WSPD pairs given  by the set $N^*(\xi)$. Since there is at most one class per WSPD pair, there are at most $\BigOT(1)$ many boundary clusters per subcell. More specifically, for every $(\Psi_1,\Psi_2) \in N^*(\subcell)$, we create a cluster,
$$A_\subcell^{(\Psi_1,\Psi_2)} = \{a_i \in A_\subcell \mid (a_i,b_i) \in M, b_i \in B_{\Psi_2} \},$$  
$$B_\subcell^{(\Psi_1,\Psi_2)} = \{b_i \in B_\subcell \mid (a_i,b_i) \in M, a_i \in A_{\Psi_2}\}.$$
\end{itemize}

For every cell $\cell'$ and any subcell $\xi \in \subcells{\cell'}$, there are a total of $\BigOT(1)$ clusters. Therefore, the total number of clusters at $\cell'$ is $\BigOT(\mu_j)$.
Let $\X_{\cell'}$ be the set of clusters at $\cell$ that are generated from its child $\cell'$. The cluster set at $\cell$ is simply $V_{\cell} = \bigcup_{\cell'\in D(\cell)}\X_{\cell'}$. We use $\mathcal{A}_{\cell}$ (resp. $\mathcal{B}_{\cell})$ to denote the vertices of type $A$ (resp. type $B$) in $V_{\cell}$; $V_{\cell} = \mathcal{A}_{\cell} \cup \mathcal{B}_{\cell}$.

Next, we partition the clusters in $\X_{\cell'}$ into two subsets called the \emph{entry} and \emph{exit} clusters respectively, 
\begin{eqnarray*}
\X_{\cell'}^\downarrow = \{B_\subcell^F, A_\subcell^I, B_\subcell^{(\Psi_1,\Psi_2)} \mid \subcell \in \subcells{\cell'}, (\Psi_1,\Psi_2) \in N^*(\subcell)\},\\
\X_{\cell'}^\uparrow = \{A_\subcell^F, B_\subcell^I, A_\subcell^{(\Psi_1,\Psi_2)} \mid \subcell \in \subcells{\cell'}, (\Psi_1,\Psi_2) \in N^*(\subcell)\}.
\end{eqnarray*}
We also denote all the clusters at $\cell$ from $\cell' \in D(\cell)$ that contains points of $A$ and $B$ as $\A_{\cell'}$ and $\B_{\cell'}$ respectively; $\X_{\cell'} = \A_{\cell'} \cup \B_{\cell'}$. 

We next describe the significance of entry and exit clusters.
For any directed path $\Pi$  in $\res{\cell^*}$, let $\overline{\pi}$ be a maximal connected sub-path of $\Pi$ that lies
inside $\cell$. 
Suppose $\overline{\pi}$ contains at least one edge. For the two endpoints $p,q$ of $\overline{\pi}$,
we refer to $p$ as entry and $q$ as exit point if $\overline{\pi}$ is directed from $p$ to $q$.
Then, it was shown in~\cite{sa_stoc12} that the entry point lies in an entry cluster and exit point lies in an exit cluster. 
\subsection{Relating parent-child clusters}

Let $\cell$ be any node in $\mathcal{T}_{\cell^*}$. 
Note that all internal nodes of the active tree except the root have four children. For any cell $\cell$ of an active tree $\mathcal{T}_{\cell^*}$, clusters are defined with respect to the subcells of its children. For any cell $\cell$ of $\mathcal{T}_{\cell^*}$, including the root, let $\xi$ be a subcell of $\cell' \in D(\cell)$ and let $\cell$ be a cell of level $i$. \ignore{Recollect that $A_\xi\cup B_\xi$ are partitioned into $\BigOT(1)$ clusters at $\cell$.   

the relationship between the clusters at $\cell$ and those at its children $\{\cell_1,\ldots, \cell_l\}$ can be summarized as follows. Each subcell $\subcell \in \Gd[\cell_j]$ is contained in one of the four children of $\cell_j$, say $\cell'$. Then, with respect to $\cell'$, either $\subcell$ is also a subcell in $\cell'$ or the  subcell $\subcell$ can be decomposed into a set $\langle \subcell_1, \ldots, \subcell_4 \rangle$ 
in $\Gd[\cell']$. If $\subcell$ is also a subcell of $\cell'$, then the clusters of $\cell'$ inside $\subcell$ are also clusters of $\cell$.  Otherwise, suppose $\langle \subcell_1, \ldots, \subcell_4 \rangle$ are the four subcells of $\cell'$ that together cover $\cell$.}  Then, we get the following relationship between clusters of $\cell$ and $\cell'$.
\begin{eqnarray}
\label{eq:subset}
A_\subcell^F &=& \bigcup_{\subcell' \in D(\subcell)}A_{\subcell'}^F,\qquad B_\subcell^F = \bigcup_{\subcell' \in D(\subcell)} B_{\subcell'}^F,\nonumber\\ 
A_\subcell^{(\Psi_1,\Psi_2)} &=& \bigcup_{\subcell' \in D(\subcell)} A_{\subcell'}^{(\Psi_1,\Psi_2)},\qquad B_{\subcell}^{(\Psi_1,\Psi_2)} = \bigcup_{\subcell' \in D(\subcell)} B_{\subcell'}^{(\Psi_1,\Psi_2)}, \nonumber\\
A_\subcell^I &=& \bigcup_{\subcell' \in D(\subcell)} (A_{\subcell'}^I \cup (\bigcup_{(\Psi,\Psi') \in N^i(\subcell')} A_{\subcell'}^{(\Psi,\Psi')})),\nonumber\\
B_\subcell^I &=& \bigcup_{\subcell' \in D(\subcell)} (B_{\subcell'}^I \cup (\bigcup_{(\Psi,\Psi') \in N^i(\subcell')} B_{\subcell'}^{(\Psi,\Psi')})).\nonumber
\end{eqnarray}

For any cluster $X$ defined at a full cell $\cell$, we use the notation $D(X)$ to denote all the clusters at the children that combine to form $X$. Note that if $X$ is a cluster generated at a subcell $\xi$ of a leaf (i.e., sparse) cell $\cell' \in D(\cell)$ of $\mathcal{T}_{\cell}$, then we set $D(X)$ to be all the points that are contained in $X$.
The following lemma whose proof is straightforward states the property of the above hierarchical clustering scheme. 
\begin{lemma}
\label{lem:sim}
For any cell $\cell \in Q$, let $\cell_1,\cell_2$ be two of its children. Let $X \in \X_{\cell_1}$ and $Y \in \X_{\cell_2}$. Then the net-costs of all edges in $X\times Y$ are the same in $\res{M}$ and all such
edges are oriented in the same direction --- either all are oriented from $B$ to $A$ or all of them are oriented from $A$ to $B$.
\end{lemma}

\begin{figure}
    \centering
    \includegraphics[width = \textwidth/2]{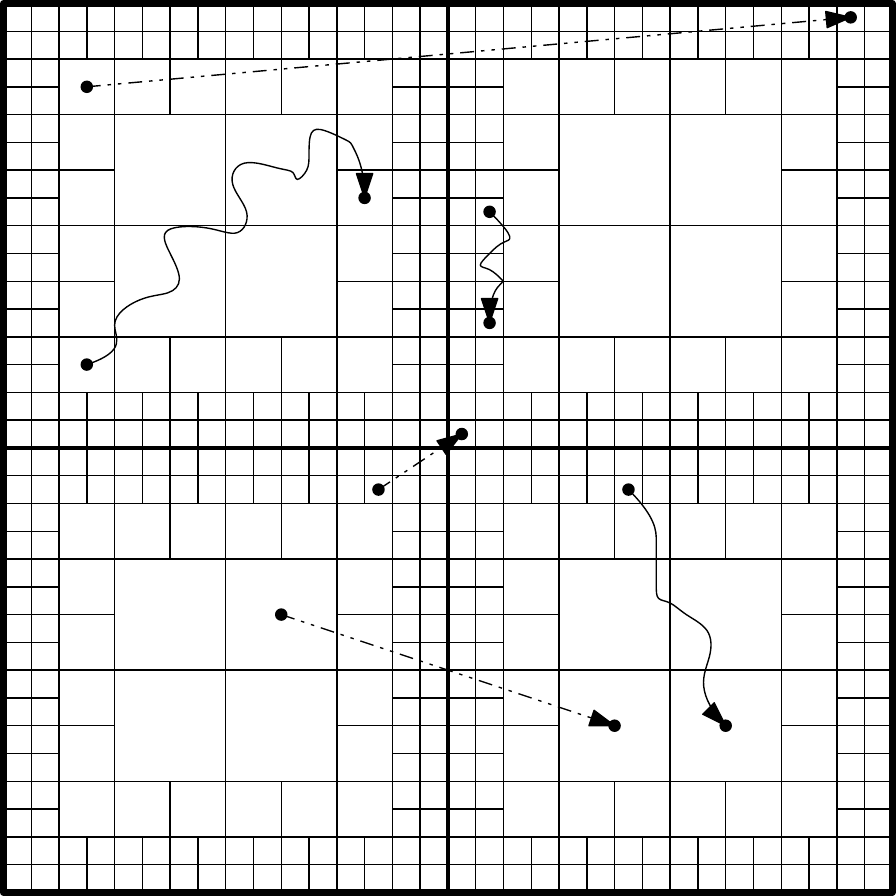}
    \caption{Examples of edges of $E_\cell$. Internal edges (solid) represent shortest paths between clusters within a child of $\cell$. Bridge edges (dashed) are between clusters in two different children of $\cell$.}
    \label{fig:associatedgraph}
\end{figure}
\subsection{Edges of the associated graph}
\label{subsec:edgesofassociatedgraph}
Given a full cell $\cell$, we already defined the vertex set $V_{\cell}$ for the associated graph.  We have the following types of edges in the edge set $E_{\cell}$ of the associated graph.
\begin{itemize}
    \item {\it Internal Edges:} For any child $\cell'$ of $\cell$, we add edges from $X$ to $Y$ provided $X \in \X_{\cell'}^\downarrow$ is an entry cluster of $\cell'$ and $Y \in \X_{\cell'}^\uparrow$ is an exit cluster of $\cell'$.
    \item {\it Bridge Edges:} For any children $\cell' \neq \cell''$ of $\cell$, for any two clusters $X$ and $Y$ where $X \in \X_{\cell'}$ and $Y \in \X_{\cell''}$, suppose $X \in \B_{\cell'}$ and $Y\in \A_{\cell''}$. We add an edge directed from $Y$ to $X$ (resp. $X$ to $Y$) if, for every edge $(x,y) \in X\times Y$, $(x,y)$ is a local (resp. non-local) edge. We continue to refer to such edges of the associated graph as local (resp. non-local) edges.
\end{itemize}

\paragraph{Bridge edge cost:} Note that for any local bridge edge from cluster $X$ to $Y$ there is at least one matching edge say $(x,y) \in X\times Y$. We set the cost of $(X,Y)$, denoted by $\phi(X,Y)$ to $\phi(x,y)$. For a non-local bridge edge $(X,Y)$, every edge $(x,y)\in (X\times Y)$ has the same net-cost, which defines the net-cost of $(X,Y)$, i.e., $\phi(X,Y) = \phi(x,y)$. Next, we describe the cost of an internal edge.

\paragraph{Internal edge costs:} For any child $\cell'$ of $\cell$, and any internal edge $(X,Y) \in (\X^\entry_{\cell} \times \X^\exit_{\cell})$ in $E_\cell$, we define its \emph{projection} $P(X,Y)$. If $\cell'$ is sparse, then $P(X,Y)$ is a minimum net-cost path in $\res{\cell'}$ from any $x \in X$ to any $y \in Y$. Otherwise, $\cell'$ is full, and the projection $P(X,Y)$ is a minimum net-cost path through $\mathcal{AG}_{\cell'}$ from any $X' \in D(X)$ to any $Y' \in D(Y)$. In either case, the net-cost of $(X,Y)$ is equal to the net-cost of its projection; i.e., $\phi(X,Y) = \phi(P(X,Y))$. The following lemma, which follows from a simple induction on the recursive definition of projection, states that any internal edge $(X,Y) \in (\X^\entry_{\cell'} \times \X^\exit_{\cell'})$ corresponds to a minimum net-cost path from $X$ to $Y$ in $\res{\cell'}$.
\begin{lemma}
\label{lem:projectfull}
For any $u,v \in \cell$ let $\Pi_{u,v,\cell}$ be a minimum net-cost alternating path in $\res{\cell}$ from $u$ to $v$.  For any internal edge $(X, Y) \in (\X_\cell^\downarrow, \X_\cell^\uparrow)$, consider $(x,y) = \argmin_{(x',y') \in X \times Y}\phi(\Pi_{x',y', \cell})$. Then $\phi(X,Y) = \phi(\Pi_{u,v,\cell})$.
\end{lemma}

\subsection{Compressed feasibility}
\label{subsec:compressedfeasibility}
Consider any active cell $\cell^*$ of the quadtree and the active tree $\mathcal{T}_{\cell^*}$ rooted at $\cell^*$. Consider an assignment of dual weights $y(\cdot)$ to the vertices of $V_{\cell}$ for all cells $\cell \in \mathcal{T}_{\cell^*}$. We say that $M_{\cell^*}$ along with these dual weights are compressed feasible if for every cell $\cell$ in $\mathcal{T}_{\cell^*}$.
\begin{itemize}
    \item[(C1)] For every edge directed from $X$ to $Y$ in $\mathcal{AG}_{\cell}$, 
    \begin{align*}
    |y(X)| - |y(Y)| &\leq \phi(X,Y),\\
    |y(X)| - |y(Y)| &= \phi(X, Y) \text{\quad if } (X, Y) \text{ is local with respect to } M_{\cell}.
\end{align*}

    \item[(C2)] If $\cell$ is full, then for each exit cluster $X \in \X^\exit_{\cell}$, for any $X' \in D(X)$,  $|y(X')| \leq |y(X)|$.
\end{itemize}
Note that if $\cell$ is sparse, then we are at a leaf node of the active tree and only condition (C1) applies. Condition (C1) implies that $M_{\cell}$ and $y(\cdot)$ are $Q$-feasible when $\cell$ is sparse. 

We define slack of any edge (bridge or internal) directed from $X$ to $Y$, denoted by $s(X,Y)$, as $\phi(X,Y)-|y(X)|+|y(Y)|$. From (C1), it follows that the slack of any edge is non-negative. We define a slack-weighted associated graph, denoted by $\mathcal{AG}_{\cell}'$, to be identical to the associated graph $\mathcal{AG}_{\cell}$, but where the weight of any edge $(X,Y)$ is its slack $s(X,Y)$.

We introduce two procedures, namely \sync\ and \construct. Both these procedures will be used to support \init, \genduals, \search\ and \augment\ operations.

\subsection{The \construct\ procedure}
\label{subsec:construct}
In this section,  we present a procedure called \construct, which will be used to compute the internal edges of an associated graph. The \construct\ procedure accepts a cell $\cell'$, such that $\cell'\neq \cell^*$ and $\mathcal{AG}_{\cell'}$ has already been computed, along with dual weights $y(\cdot)$ for all vertices of $V_{\cell'}$. It assumes that $M_{\cell'}, y(\cdot)$ satisfy the compressed feasibility conditions. Let $\cell$ be the parent of $\cell'$ in $\mathcal{T}_{\cell^*}$. The procedure computes the internal edges of $\X^\entry_{\cell'} \times \X^\exit_{\cell'}$ in $\mathcal{AG}_{\cell}$. It also assigns dual weights to the vertices of $V_{\cell}$ that correspond to clusters generated for the subcells of $\cell'$.

We describe the process for building the internal edges going out of each cluster $X \in\X^\entry_{\cell'}$. We add an additional vertex $s$ to $\mathcal{AG}_{\cell'}'$ and add an edge from $s$ to each cluster $X' \in D(X)$ with a cost equal to $|y(X')|$. After creating this \emph{augmented associated graph}, we simply execute Dijkstra's algorithm from $s$ to find the shortest path distance from $s$ to every node in $V_{\cell'}$. Let $d_v$ denote the shortest path distance from $s$ to $v$ in the augmented associated graph. For each exit cluster $Y \in \X^\exit_{\cell'}$, we create an internal edge from $X$ to $Y$ in $\mathcal{AG}_{\cell}$ and set its cost to be $\min_{Y'\in D(Y)} (d_{Y'} - |y(Y')|)$. We repeat this procedure for each entry cluster.      

This completes the description how to construct the internal edges of $\mathcal{AG}_\cell$ in $\cell'$. We next assign dual weights to each cluster $X \in \X_{\cell'}$ as follows: If $X$ is an entry cluster, let $X' = \argmin_{Y \in D(X)}|y(Y)|$. Otherwise, $X$ is an exit cluster, and we let $X' = \argmax_{Y \in D(X)}|y(Y)|$. In either case, we set $y(X) \leftarrow y(X')$.

The following lemma shows that the \construct\ procedure correctly assigns the net-cost of edges in $\X^\entry_{\cell'} \times \X^\exit_{\cell'}$.
\begin{lemma}
\label{lem:correctcost}
Let $X \in \X^\entry_{\cell'}$ and $Y \in \X^\exit_{\cell'}$ be a pair of clusters of $\cell'$ that form an internal edge $(X,Y) \in E_{\cell}$. Then the \construct\ procedure ensures that $\phi(X,Y) = \phi(P(X,Y)))$.
\end{lemma}
\begin{proof}
Consider any $X' \in D(X)$ and $Y' \in D(Y)$; let $P_{X',Y'}$ be the minimum net-cost path from $X'$ to $Y'$ in $\mathcal{AG}_{\cell'}$. For any edge $(u,v) \in P_{X',Y'}$, from the definition of slack, we have $|y(u)| - |y(v)| + s(u,v) = \phi(u,v)$. When summing over all $(u,v) \in P_{X',Y'}$, we get that each vertex of the path except the first or last vertex has a net-contribution of $0$ to the dual weight magnitude total. From the definition of projection,
\begin{equation}
\label{eq:slacknetcostprojection}
|y(X')| - |y(Y')| + \sum_{(u,v) \in P_{X',Y'}} s(u,v) = \sum_{(u,v) \in P_{X',Y'}} \phi(u,v).
\end{equation}

When executing a Dijkstra search from $X$, the \construct\ procedure assigns,

\begin{align*}
    \phi(X,Y) &= \min_{Y' \in D(Y)} d_{Y'} - |y(Y')| \\
    &=  \min_{X' \in D(X), Y' \in D(Y)}|y(X')| - |y(Y')| + \sum_{(u,v) \in P_{X',Y'}} s(u,v)\\
    &= \min_{X' \in D(X), Y' \in D(Y)}\sum_{(u,v) \in P_{X',Y'}}\phi(u,v)\\
    &= \phi(P(X,Y)).
\end{align*}
Therefore, the \construct\ procedure correctly computes $\phi(X,Y)$.
\end{proof}
In the following Lemma, we argue that the internal edges of $\cell'$ in $\mathcal{AG}_{\cell}$ are feasible after \construct\ is called on $\cell'$.
\begin{lemma}
\label{lem:internalfeasconstruct}
After \construct\ is called on a cell $\cell'$ with parent $\cell$, then, for every internal edge $(X,Y) \in \X^\entry_{\cell'} \times \X^\exit_{\cell'}$ of $E_\cell$, we have, $|y(X)| - |y(Y)| \leq \phi(X,Y)$.
\end{lemma}
\begin{proof}
The edge $(X,Y)$ has some projection $P(X,Y)$; let $X' \in D(X)$ (resp. $Y' \in D(Y)$) be the first (resp. last) vertex of $P(X,Y)$. From equation \eqref{eq:slacknetcostprojection} and the feasibility of $\mathcal{AG}_{\cell'}$, we have,
\[|y(X')| - |y(Y')| \leq \phi(P(X,Y)) = \phi(X,Y).\]
From the way the \construct\ procedure assigns dual weights to $X$ and $Y$, we have that $|y(X)| \leq |y(X')|$ and $|y(Y')| \leq |y(Y)|$. Therefore, $|y(X)| - |y(Y)| \leq \phi(X,Y)$.
\end{proof}

\paragraph{Efficiency of \construct:}
Next, we bound the time taken for a single call to \construct\ on a cell $\cell \in \mathcal{T}_{\cell^*}$. Assume that $\cell$ appears at level $j$. The \construct\ procedure executes a Dijkstra search from each of the $|\X_\cell| = \BigOT(\mu_j)$ clusters of $\cell$. If $\cell$ is full, then each Dijkstra search takes $\BigOT(|E_{\cell}|) = \BigOT(|V_{\cell}|^2) = \BigOT(\mu_j^2)$ time. If $\cell$ is sparse, then each Dijkstra search can be executed efficiently in $\BigOT(|\X_\cell| + |A_{\cell} \cup B_{\cell}|)$ time using the fact that the edges of $\res{M}$ outgoing from any vertex $v$ belong to only $\BigOT(1)$ different WPSD pairs; the same technique was used for the Hungarian search in Section \ref{sec:algorithm}. Since $\cell$ is sparse, $|A_\cell \cup B_\cell| \leq \mu_j^2$, and each Dijkstra search takes $\BigOT(\mu_j^2)$ time. The \construct\ procedure executes $|\X^\entry_\cell|=\BigOT(\mu_j)$ Dijkstra searches, and each Dijkstra search takes $\BigOT(\mu_j^2)$ time, so the total time taken by \construct\ is $\BigOT(\mu_j^3)$ for any cell of level $j$. This gives the following Lemma.
\begin{lemma}
\label{lem:constructefficiency}
Any execution of \construct\ on a cell $\cell$ of layer $j$ takes $\BigOT(\mu_j^3)$ time. Furthermore, if $\cell$ is sparse, the time taken can be bounded by $\BigOT(\mu_j(\mu_j + |A_{\cell}| + |B_{\cell}|))$.
\end{lemma}

\ignore{
\begin{lemma}
Suppose we are given a compressed feasible matching $M$ and $y(\cdot)$ in an active cell $\cell^* \in G_i$. Now suppose we update the matching $M$ to $M'$ and the dual assignments from $y(\cdot)$ to $y'(\cdot)$ with the following properties. For any cell $\cell \in \mathcal{T}_{\cell^*}$ and $\cell' \in D(\cell)$, every vertex $v \in V_{\cell'}$ has a dual assignment $y(v)$ such that
\begin{itemize}
    \item  $y'(v) \le y(v)$, and, 
    \item Every local bridge edge and internal edge of $\mathcal{AG}_{\cell'}$ satisfies compressed feasibility condition (C1). 
\end{itemize}
Then, applying construct for $\cell'$ will produce an a set of internal edges of $\cell'$ in $E_\cell$ and a dual assignment for their end points in $V_{\cell}$ such that (C1) and (C2) are satisfied.
\end{lemma}
}

\subsection{\sync\ Procedure}
\label{subsec:sync}
For an active cell $\cell^*$ and a compressed feasible matching $M_{\cell^*}$ along with a set of dual weights $y(\cdot)$, the \sync\ procedure takes the updated dual weights on clusters of $\X_{\cell}$ at any non-root cell $\cell \in \mathcal{T}_{\cell^*}$ and uses them to update the dual weights of $V_{\cell}$ such that the matching continues to be compressed feasible, and,
\begin{enumerate}[(T1)]
    \item For any entry cluster $X \in \X^\entry_\cell$, and for any $X' \in D(X)$, $|y(X')| \geq |y(X)|$.
    \item For any free or boundary cluster $X \in \X_\cell$, and for any $X' \in D(X)$, $y(X') = y(X)$.
\end{enumerate}

The \sync\ procedure consists of executing the following algorithm for each entry cluster $X \in \X_{\cell}^{\entry}$:
We create a new vertex $s$ and add an edge from $s$ to each vertex $X' \in D(X)$. We assign a weight $|y(X')|$ to the edge from $s$ to $X'$. Then, we execute Dijkstra's algorithm starting from $s$. Let $d_v$ be the shortest path distance from $s$ to $v$ as computed by Dijkstra's algorithm. For any vertex $v$ with $d_v < |y(X)|$, if $v \in \mathcal{B}_{\cell}$ we update the dual weight to $y(v) \leftarrow y(v)+|y(X)|-d_v$. Otherwise, $v \in \mathcal{A}_{\cell}$, and we update the dual weight $y(v) \leftarrow y(v) - |y(X)| + d_v$. Note that in both cases the magnitude of the dual weight increases by $|y(X)|-d_v$. The dual weight of every other vertex with $d_v \ge |y(X)|$ does not change. This completes the description of the algorithm initiated with respect to $X$.

Note that for any cluster $X' \in D(X)$ if  $|y(X')| \ge |y(X)|$, then the procedure will only further increase the magnitude of $y(X')$ and so, (T1) holds. If, on the other hand, $|y(X')| < |y(X)|$, then the length of the edge from $s$ to $X'$ is $|y(X')|$, and so the shortest path distance $d_{X'}\le |y(X')| < |y(X)|$. The magnitude of the dual weight of $X'$ increases by $|y(X)| - d_{X'} \ge |y(X)|- |y(X')|$ implying that the new magnitude of $y(X')$ is at least the magnitude of $y(X)$. Therefore (T1) holds for any entry cluster.  

After we execute this for all entry clusters, we perform the following dual adjustment: For any $X \in \X_{\cell}$, and for any $X' \in D(X)$, we will explicitly update the dual weight $y(X')$ to match $y(X)$. Therefore, (T2) holds after the execution of \sync.

To prove the correctness of \sync\ it remains to show that the updated dual weights satisfy compressed feasibility.

\begin{lemma}
\label{lem:syncproof}
For any compressed feasible matching $M_{\cell^*}$ on an active cell $\cell^*$, after the execution of $\sync$ at $\cell \in \mathcal{T}_{\cell^*}$ the updated dual weights of $V_{\cell}$ will continue to satisfy (C1) and (C2).
\end{lemma}
\begin{proof}
Assume that the claim holds prior to executing Dijkstra's algorithm from some vertex $X \in \X^\entry_\cell$. We argue that (C1) and (C2) continue to hold after executing the algorithm from $X$. Let $y(\cdot)$ (resp. $y'(\cdot)$) be the dual weights after (resp. before) executing this process with respect to $X$. Let $s(u,v)$ be the slacks with respect to $y(\cdot)$.  

We begin by arguing that if $(u,v)$ is local, and $d_u \neq \infty$, then $d_u = d_v$. First, consider if $\cell$ is full. Then, the only incoming edge to $v$ is via $u$. Therefore, if $u$ is reached during the Dijkstra search, $d_u = d_v$. Next, consider if $\cell$ is sparse. Then, if $(u,v) \in M$, $(u,v)$ is the only edge incoming to $v$, so clearly if $u$ is reached during the search, then $d_u = d_v$. Next, consider if $(u,v) \notin M$. Then, since $(u,v)$ is local, $u \in B$ must be matched to some vertex $v' \in A$ and $v \in A$ must be matched to some vertex $u' \in B$. Furthermore, there must be a non-matching local edge from $u'$ to $v'$. If $u$ is reached, then, since $(u,v)$ has $0$ slack, $d_v \leq d_u$. Furthermore, there is a path $\langle v, u', v', u\rangle$ from $v$ to $u$ consisting solely of local edges, which all have $0$ slack. Therefore, $d_u \leq d_v$, and we get that $d_u = d_v$.   

Next, we show (C1). Consider any edge $(u,v) \in \mathcal{AG}_{\cell}$. We consider four cases:
\begin{itemize}
    \item $d_u, d_v \geq |y(X)|$ : In this case, the dual weights of $u$ and $v$ are unchanged, so the edge remains feasible with respect to (C1).
    \item $d_u, d_v < |y(X)|$ : We first observe that, from the definition of shortest paths, $d_v - d_u \leq s(u,v)$. If $(u,v)$ is non-local, then,
    \begin{align*}
        |y'(u)| - |y'(v)| &= |y(u)| + (|y(X)| - d_u) - (|y(v)| + (|y(X)| - d_v))\\
        &= |y(u)| - |y(v)| + d_v - d_u\\
        &\leq |y(u)| - |y(v)| + s(u,v) \\
        &= \phi(u,v),
    \end{align*}
    and the edge $(u,v)$ satisfies (C1). Furthermore, if $(u,v)$ is local, then $d_v - d_u = 0 = s(u,v)$, and the same equation holds with equality, i.e., $|y'(u)| - |y'(v)| = \phi(u,v)$, which satisfies (C1).
    \item $d_u < |y(X)|$ and $d_v \geq |y(X)|$ : Since $d_u \neq d_v$, $(u,v)$ is non-local. We have:
    \begin{align*}
        |y'(u)| - |y'(v)| &= |y(u)| + (|y(X)| - d_u) - |y(v)|\\
        &\leq |y(u)| - |y(v)| + d_v - d_u\\
        &\leq |y(u)| - |y(v)| + s(u,v) \\
        &= \phi(u,v),
    \end{align*}
    and (C1) is satisfied.
    \item $d_u \geq |y(X)|$ and $d_v < |y(X)|$ : Since $d_u \neq d_v$, $(u,v)$ is non-local. We have:
    \begin{align*}
        |y'(u)| - |y'(v)| &= |y(u)| - (|y(v)| + (|y(X)| - d_v))\\
        &\leq |y(u)| - |y(v)|\\
        &\leq \phi(u,v),
    \end{align*}
    and (C1) is satisfied.
\end{itemize}
We conclude that (C1) is satisfied after executing the algorithm from $X$. Next, we argue (C2) continues to hold after executing the algorithm from $X$. Consider any exit cluster $Y \in \X^\exit_{\cell}$ and any $Y' \in D(Y)$. 
If $d_{Y'} \geq |y(X)|$, the claim trivially holds because $y(Y')$ did not change. Otherwise, consider the shortest path $P_{Y'}$ from $X$ to $Y'$ computed by Dijkstra's algorithm. Let $P_{X',Y'}$ be the path $P_{Y'}$ with the vertex $X$ removed; $X'$ is the first vertex after $X$ on $P_{Y'}$. 
From the feasibility of $(X,Y)$, and \eqref{eq:slacknetcostprojection} we have,
\begin{align*}
    |y(X)| - |y(Y)| &\leq \phi(X,Y)\\
                    &\leq \sum_{(u,v) \in P_{X',Y'}}\phi(u,v)\\
                    &= |y(X')| - |y(Y')| + \sum_{(u,v) \in P_{X',Y'}}s(u,v)\\
                    &= d_{Y'} - |y(Y')|.
\end{align*}
Combining this with the dual weight assignment of the procedure gives:
\[|y'(Y')| = |y(Y')| + |y(X)| - d_{Y'}\\
    \leq |y(Y)|,\]
and (C2) holds.

Finally, we argue that the final step of \sync, which assigns the dual weights of some free and boundary clusters to match their parent cluster's dual weight, does not violate (C1); it is clear that this operation does not violate (C2). Consider any free or boundary cluster $X \in \X_{\cell}$, and any child cluster $X' \in D(X)$. If $X$ is an exit cluster, then $X'$ is of type $A$, and there are no outgoing edges from $X'$ in $V_{\cell}$. For any incoming edge $(Z, X')$, the slack only increases because, from (C2), $|y(X')|$ only increases. Similarly, if $X$ is an entry cluster, then $X'$ is of type $B$, and there are no incoming edges to $X'$ in $V_{\cell}$. For any outgoing edge $(X', Z)$, the slack only increases because, from (T1), $|y(X')|$ only decreases. Therefore, the final step of \sync\ does not violate (C1) or (C2).
\end{proof}

\begin{lemma}
\label{lem:auguse}
Consider any internal edge $(X,Y) \in \X^\entry_{\cell} \times \X^\exit_{\cell}$ in the associated graph $\mathcal{AG}_{\hat{\cell}}$ of the parent $\hat{\cell}$ of $\cell$. Suppose $s(X,Y)$ is $0$. After execution of \sync\ on $\cell$, $P(X,Y)$ is an admissible path in $\mathcal{AG}_{\cell}$.
\end{lemma}
\begin{proof}
By its definition, $P(X,Y)$ is a path from some $X' \in D(X)$ to some $Y' \in D(Y)$. Note that the net-cost $\phi(X,Y) = \phi(P(X,Y))$. From our assumption that $(X,Y)$ is admissible and \eqref{eq:slacknetcostprojection},
\[
    |y(X)|-|y(Y)|=\phi(X,Y) = \phi(P(X,Y)) = |y(X')| - |y(Y')| + \sum_{(u,v) \in P(X,Y)}s(u,v).
\]

Since $X$ is an entry cluster, from (T1), $|y(X')| \ge |y(X)|$.  Since $Y$ is an exit cluster, from (C2), it follows that $|y(Y')| \le |y(Y)|$. Therefore,
\[\sum_{(u,v) \in P(X,Y)}s(u,v) = |y(X)| - |y(X')| + |y(Y')| - |y(Y)| \leq 0.\]
From Lemma \ref{lem:syncproof}, the edges of $P(X,Y)$ satisfy (C1), and every edge of $P(X,Y)$ has slack at least $0$. Therefore, every edge of $P(X,Y)$ must be admissible.
\end{proof}

By recursively applying the above lemma, we get the following.

\begin{cor}
\label{cor:auguse}
Let $\cell$ be a level $i$ cell. For any internal edge $(X,Y) \in \X^\entry_{\cell} \times \X^\exit_{\cell}$ in the associated graph $\mathcal{AG}_{\hat{\cell}}$ of the parent $\hat{\cell}$ of $\cell$, suppose $s(X,Y)$ is $0$. We can recursively apply \sync\ on all internal edges of $P(X,Y)$ to obtain its projection $\Pi_{u,v,\cell}$ with $u\in X$ and $v \in Y$. This projection will be an admissible path. For every vertex $p$ in $\Pi_{u,v,\cell}$, let $\mathbb{P}(p)$ be all the clusters for cells of level $i$ or lower that contain the point $p$. Then, for every $v' \in \mathbb{P}(p)$, $y(v') = y(p)$. 
\end{cor}

\paragraph{Efficiency Analysis of \sync:}
Next, we bound the time taken for a single call to \sync\ executed on a cell $\cell$ that updates the dual weights of $V_{\cell}$. The argument is nearly identical to that used for \construct. Assume that $\cell$ appears at level $j$. The \sync\ procedure executes a Dijkstra search once from each of the $\BigOT(\mu_j)$ entry clusters of $\X_{\cell}$. If $\cell$ is full, then each Dijkstra search takes time $\BigOT(|E_{\cell}|) = \BigOT(|V_{\cell}|^2) = \BigOT(\mu_j^2)$ time. If $\cell$ is sparse, then each Dijkstra search can be executed efficiently in $\BigOT(|\X_\cell| + |A_{\cell} \cup B_{\cell}|)$ time. Since $\cell$ is sparse, $|A_\cell \cup B_\cell| \leq \mu_j^2$, and each Dijkstra search takes $\BigOT(\mu_j^2)$ time. This gives the following Lemma.
\begin{lemma}
\label{lem:syncefficiency}
Any execution of \sync\ on a cell $\cell$ of layer $j$ takes $\BigOT(\mu_j^3)$ time. Furthermore, if $\cell$ is sparse, the time taken can be bounded by $\BigOT(\mu_j(\mu_j + |A_{\cell}| + |B_{\cell}|))$.
\end{lemma}

\subsection{Data Structure Operations}
For any phase $i$, we present the implementation of the four operations supported by the data structure using the \sync\ and \construct\ procedures. Before we describe the operations, we will state an additional property that the compressed feasible matching maintained by the data structure satisfies. In any phase $i \ge 1$, suppose that $M_{\cell^*}$, $y(\cdot)$ is a compressed feasible matching with the additional condition being satisfied: 
\begin{itemize}
    \item[(J)] For each vertex $b \in \mathcal{B}_{\cell^*}$, $y(b) \ge 0$, and for each $a \in  \mathcal{A}_{\cell^*}$, $y(a) \le 0$. Furthermore, let $y_{\max}=\max_{v \in \mathcal{B}_{\cell^*}}y(v)$. For every free vertex $b \in \mathcal{B}_{\cell^*}$, $y(b)=y_{\max}$ and $\mu_{i-1}^2 \le y_{\max} \le \mu_i^2$. For every free cluster $a \in \mathcal{A}_{\cell^*}$, $y(a) =0$.
\end{itemize} 

As we show in Section~\ref{subsec:genduals}, a compressed feasible matching that satisfies (J) can be converted to an associated $Q$-feasible matching that satisfies (I1) and (I2). Therefore, it suffices to maintain (J) during the execution of our algorithm.

\subsubsection{\init\ Operation}
\label{subsec:build}
As input, the \init\ operation takes a $Q$-feasible matching $M_{\cell^*}$ and set of dual weights $y(\cdot)$ on the vertices of $A_{\cell^*} \cup B_{\cell^*}$. We execute the \construct\ procedure on every non-root cell $\cell$ of $\mathcal{T}_{\cell^*}$ in the order of their level in $Q$, processing lower layers first. This ensures that, when \construct\ is called on $\cell$, the associated graph $\mathcal{AG}_{\cell}$ has already been computed, along with the dual weights for vertices of $V_\cell$. After \construct\ is called on all pieces of $\cell^*$, the result is an associated graph $\mathcal{AG}_{\cell}$ for every full cell $\cell \in \mathcal{T}_{\cell^*}$ and dual weights $y(\cdot)$ for all vertices of $\bigcup_{\cell \in \mathcal{T}_{\cell^*}}V_{\cell}$. The following lemma argues that this set of dual weights is compressed feasible with respect to $M_{\cell^*}$.

\begin{lemma}
\label{lem:constructfull}
After executing \construct\ on all non-root cells of $\mathcal{AG}_{\cell^*}$, the matching $M_{\cell^*}$ and the dual assignment $y(\cdot)$ are compressed feasible.
\end{lemma}
\begin{proof}
Consider the circumstances after calling \construct\ on all children of some full $\cell \in \mathcal{T}_{\cell}$. Inductively assume that, for each child $\cell'$ of $\cell$, the edges of $\mathcal{AG}_{\cell'}$ were feasible prior to executing \construct\ on $\cell'$. Then, from Lemma \ref{lem:internalfeasconstruct}, the internal edges of $\mathcal{AG}_{\cell}$ are feasible. It remains to argue that the bridge edges of $\mathcal{AG}_{\cell}$ are feasible. Consider any such bridge edge $(X,Y)$ in $\mathcal{AG}_{\cell}$. From a simple inductive argument on the dual assignment of the \construct\ procedure, it is easy to see that for some point $u \in X$, $y(u) = y(X)$. Similarly, for some point $v \in Y$, $y(v) = y(Y)$. Furthermore, $(u,v)$ is an edge in $\res{M}$, and $\phi(X,Y) = \phi(u,v)$. Consider the case where $(X,Y)$ is local. Then, from the feasibility of $(u,v)$, we have,
\[|y(X)| - |y(Y)| = |y(u)| - |y(v)| = \phi(u,v) = \phi(X,Y),\]
and $(X,Y)$ is feasible. Similarly, consider if $(X,Y)$ is non-local. Then, from the feasibility of $(u,v)$, we have,
\[|y(X)| - |y(Y)| = |y(u)| - |y(v)| \leq \phi(u,v) = \phi(X,Y),\]
and $(X,Y)$ is feasible. This implies that $M_{\cell^*}, y(\cdot)$ are compressed feasible after executing \construct\ on all non-root cells of the active tree $\mathcal{T}_{\cell^*}$.
\end{proof}
\begin{cor}
Given a $Q$-feasible matching that satisfies (I1) and (I2), upon applying the \init\ procedure, the compressed feasible matching will satisfy (J).
\end{cor}
\begin{proof}
Consider any cluster $X \in V_{\cell*}$. Then, there must be some $u \in X$ for which $y(u) = y(X)$ after \init. Since (I1) and (I2) were satisfied for $u$, it is easy to see that (J) holds for $X$.
\end{proof}
\paragraph{Execution Time for \init:}
We show that the time taken by \init\ during phase $i$ is $\BigOT(n\mu_i^{2/3})$. During, the \init\ procedure, \construct\ is called on all non-root cells of $\mathcal{T}_{\cell^*}$ for each full active cell $\cell^*$. We assign each non-root cell $\cell \in \mathcal{T}_{\cell^*}$ to one of four categories: 
 \begin{enumerate}[(a)]
     \item $\cell$ is full.
     \item $\cell$ is sparse and the parent of $\cell$ in $\mathcal{T}_{\cell^*}$ is a full cell that is not the root $\cell^*$.
     \item $\cell$ is sparse, its parent in $\mathcal{T}_{\cell^*}$ is the root $\cell^*$, and $|A_\cell \cup B_\cell| \leq \mu_i^{2/3}$.
     \item $\cell$ is sparse, its parent in $\mathcal{T}_{\cell^*}$ is the root $\cell^*$, and $|A_\cell \cup B_\cell| > \mu_i^{2/3}$.
 \end{enumerate}
We separately bound the total time taken for a single \construct\ call on every cell in each of the four categories, over all active cells for phase $i$, showing that the time taken is $\BigOT(n\mu_i^{2/3})$. 

First, we bound the time taken by cells of category (a). We bound the the time for a \construct\ call on all full cells in some grid $G_j$. Since these full cells together contain at most $n$ points, the total number of full cells in $G_j$ is bounded by $\BigOT(n/\mu_j^2)$. A \construct\ call on a full cell $\cell \in G_j$ takes $\BigOT(\mu_j^{3})$ time. Therefore, the total time taken for all full cells of $G_j$ is $\BigOT(n\mu_j)$. During phase $i$, \construct\ is only called on cells of $G_j$ where $j \leq \lfloor 2i/3 \rfloor$. The time taken by $G_{\lfloor 2i/3 \rfloor}$ dominates, taking $\BigOT(n\mu_i^{2/3})$ time. This completes the bound on cells in category (a).

Next, we bound the time taken for category (b). If a sparse cell $\cell$ of level $j$ in an active tree has a non-root parent $\cell'$ in level $j+1$, then its parent $\cell'$ must fall into category (a). The time taken for a call to \construct\ on $\cell$ is $\BigOT(\mu_j^3) = \BigOT(\mu_{j+1}^3)$, which can be taxed on the time taken to execute \construct\ on the parent $\cell'$. Specifically, since each non-root cell $\cell'$ in the active tree has at most $4$ children in the active tree, the time taken for a \construct\ call on all sparse children of $\cell'$ is $\BigOT(\mu_{j+1}^3)$, which is also the bound on the time taken for \construct\ on $\cell'$ itself. Therefore, the total time taken by category (b) is bounded by the time taken by (a), and is $\BigOT(n\mu_i^{2/3})$.

Now we bound the time for category (c). All cells of category (c) are pieces of some active tree, so we begin by bounding the total number of pieces over all active cells. Since these cells together contain at most $n$ points, the number of full active cells during phase $i$ is $\BigOT(n/\mu_i^2)$. Each such active cell $\cell^*$ has its pieces in grid $G_{\lfloor 2i/3\rfloor}$. Since $\cell^*$ has diameter $\BigOT(\mu_i^2)$ and each piece of $\cell^*$ has diameter $\Omega(\mu^2_{\lfloor 2i/3\rfloor} / \polys) = \Omega(\mu_i^{4/3}/\polys)$, $\cell^*$ has $\BigOT((\mu_i^{2/3})^2) = \BigOT(\mu_i^{4/3})$ pieces. Summing over all phase $i$ active full cells gives a total of $\BigOT(n/\mu_i^{2/3})$ pieces. The time taken by a single \construct\ call on one of these pieces $\cell$ is $\BigOT(\mu_{\lfloor 2i/3 \rfloor}|A_{\cell} \cup B_{\cell}| + \mu_{\lfloor 2i/3 \rfloor}^2) = \BigOT(\mu_{i}^{2/3}|A_{\cell} \cup B_{\cell}| + \mu_{i}^{4/3})$. However, since $|A_\cell \cup B_\cell| \leq \mu_i^{2/3}$, we can rewrite the time taken for a single \construct\ call as $\BigOT(\mu_i^{4/3})$. Summing over all $\BigOT(n/\mu_i^{2/3})$ pieces of category (c) gives a total time of $\BigOT(n\mu_i^{2/3})$ as desired.

Finally, we bound the time taken for category (d). The time taken for a single \construct\ call on a cell $\cell$ of category (d) is $\BigOT(|A_\cell \cup B_\cell|\mu_i^{2/3} + \mu_i^{4/3})$. However, since $|A_\cell \cup B_\cell| > \mu_i^{2/3}$, the first term dominates, and we can rewrite the time taken by $\construct$ on $\cell$ as $\BigOT(|A_{\cell} \cup B_{\cell}|\mu_i^{2/3})$. Summing over all such cells of category (d) gives a total time of $\BigOT(n\mu_i^{2/3})$.

\subsubsection{\genduals\ Operation}
\label{subsec:genduals}
The \genduals\ procedure simply consists of recursively calling the \sync\ procedure on all non-root cells of $\mathcal{\cell^*}$, processing cells closest to the root of $\mathcal{T}_{\cell^*}$ first. This process generates a set of dual weights $y(\cdot)$ for the vertices of $A_{\cell^*} \cup B_{\cell^*}$. Next, we show that after executing this \genduals\ procedure, $M_{\cell^*}, y(\cdot)$ are $Q$-feasible, meaning (I1) holds.

\begin{lemma}
\label{lem:syncfull}
After executing \sync\ on all non-root cells of $\mathcal{T}_{\cell^*}$, starting with the cells closest of the root of $\mathcal{T}_{\cell^*}$, let $y(\cdot)$ be the dual weights of vertices in $\bigcup_{\cell \in \mathcal{T}_{\cell^*}} V_{\cell}$. Then $M_{\cell^*}, y(\cdot)$ are $Q$-feasible.
\end{lemma}
\begin{proof}

We consider any edge $(u,v)$ in $\res{\cell^*}$. The vertices $u$ and $v$ must each appear in some sparse cell in $\mathcal{T_{\cell^*}}$. If $u$ and $v$ are in the same sparse leaf cell $\cell \in \mathcal{T}_{\cell^*}$, then $(u,v) \in \res{\cell}$, and $(u,v)$ is $Q$-feasible because $M_\cell, y(\cdot)$ are compressed feasible. Otherwise, $u$ and $v$ appear at different leaves of $\mathcal{T}_{\cell^*}$, and there is some pair of clusters $(X,Y)$ and some cell $\cell \in \mathcal{T}_{\cell^*}$ such that $(X,Y)$ is a bridge edge in $E_{\cell}$, $u \in X$, and $v \in Y$. If $(X,Y)$ is local, then $(u,v)$ is also local, and we have that $|y(X)| - |y(Y)| = \phi(X,Y)$ from (C1). From (T2) we get that $y(u) = y(X)$ and $y(v) = y(Y)$. Since $\phi(X,Y) = \phi(u,v)$, we get that $(u,v)$ is $Q$-feasible. 

Otherwise, we consider the case where $(u,v)$ is non-local. We claim that $|y(u)| \leq |y(X)|$ (resp. $|y(v)| \geq |y(Y)|$). If $X$ (resp. $Y$) is a boundary or free cluster, then, from (T2), $y(u) = y(X)$ (resp. $y(v) = y(Y)$) and the claim holds. The remaining case is when $X$ (resp. $Y$) is an internal cluster. Since the edge $(X,Y)$ is directed from $X$ to $Y$, the cluster $X$ (resp. $Y$) must be an internal exit cluster of $B$ (resp. internal entry cluster of $A$). From conditions (C2) (resp. (T1)), we immediately get that $|y(u)| \leq |y(X)|$ (resp. $|y(v)| \geq |y(Y)|$). Therefore, the claim holds. From condition (C1), $|y(X)| - |y(Y)| \leq \phi(X,Y)$. Combining this with the facts that $|y(u)| \leq |y(X)|$ and $|y(v)| \geq |y(Y)|$ immediately implies that $(u,v)$ is $Q$-feasible.
\end{proof}

\begin{lemma}
\label{compressedinv}
Consider any compressed feasible matching that satisfies (J), then \genduals\ generates an associated $Q$-feasible matching that satisfies invariant (I2). 
\end{lemma}
\begin{proof}
Consider any free vertex $v \in A_F \cup B_F$ of $\res{\cell^*}$. Then there must be some free vertex cluster $X \in \X_{\cell}$ for some piece $\cell$ of $\cell^*$ in $\mathcal{T}_{\cell^*}$, such that $v \in X$. From (T2), $y(v) = y(X)$ after \genduals. Since (J) holds for $X$, (I2) holds for $v$. Next, consider the case where $v \in A \cup B$ is not free. The \genduals\ procedure only modifies the dual weight of $v$ via the \sync\ procedure, which ensures that $y(v)$ remains non-positive (resp. non-negative) if $v \in A$ (resp. $v \in B)$. 
\end{proof}

\paragraph{Execution Time for \genduals:}
Since \construct\ and \sync\ have the same time bounds from Lemmas \ref{lem:constructefficiency} and \ref{lem:syncefficiency}, and both procedures are called on all non-root cells of the active tree $\mathcal{AG}_{\cell^*}$, it is easy to see that the time taken by \genduals\ can be bounded in a fashion identical to the argument used for the efficiency of \init. Therefore, \genduals\ takes $\BigOT(n\mu_i^{2/3})$ time for phase $i$.

\subsubsection{\search\ Operation}
\label{subsec:hungariansearch}
This procedure takes a compressed feasible matching $M_{\cell^*},y(\cdot)$ that also satisfies (J) as input. It then conducts a search identical to Hungarian search on the associated graph of $\cell^*$. The search procedure adjusts the dual weights of the vertices of $V_{\cell^*}$ so that we have a path consisting of admissible edges. Once an admissible path is found in $\mathcal{AG}_{\cell^*}$, the procedure projects this path to find an augmenting path of admissible edges in $\res{\cell^*}$ by recursively applying the \sync\ procedure. We describe the details of the procedure in two parts. First, we describe the dual adjustments conducted by the \search, and then we describe how the procedure projects the path. We show that (J) continues to hold after the execution of \search.

\paragraph{Dual Adjustments:} Recall that $\mathcal{A}_{\cell^*}$ (resp. $\mathcal{B}_{\cell^*}$) denotes the set of vertices of type $A$ (resp. type $B$) in $V_{\cell^*}$. Let $\mathcal{A}_F$ (resp. $\mathcal{B}_{F}$) be the set of free vertex clusters of $\mathcal{A}_{\cell^*}$ (resp. $\mathcal{B}_{\cell^*}$). We add a vertex $s$ to the graph $\mathcal{AG}_{\cell^*}'$ and add an edge from $s$ to every free cluster of $\mathcal{B}_F$. The weight associated with this edge is $0$. We set $\ell_{\max}= \mu_i^2-\max_{X\in \mathcal{B}_{\cell^*}} y(X)$.
We then execute a Dijkstra's search to compute the shortest path distance from $s$ to every vertex in $V_{\cell^*}$. For any  $v \in V_{\cell^*}$, let $\ell_v$ be the shortest path distance from $s$. Let $\ell = \min_{X\in \mathcal{A}_F} \ell_X$. If $\ell > \ell_{\max}$, we set $\ell = \ell_{\max}$ and continue. For every vertex $v \in V_{\cell^*}$ with $\ell_v \le \ell$, we update the dual weight as follows. If $v \in \mathcal{B}_{\cell^*}$, we increase the dual weight $y(v) \leftarrow y(v) + \ell - \ell_v$. Otherwise, if $v \in \mathcal{A}_{\cell^*}$, we reduce the dual weight $y(v) \leftarrow y(v) - \ell + \ell_v$. This completes the description of the dual weight changes. These dual adjustments will make some of the edges on the shortest path tree have a zero slack. If $\ell = \ell_{\max}$, the dual weight of every free cluster of type $B$ would be updated to $\mu_i^2$ and we return without finding an augmenting path. Otherwise, the dual adjustments will maintain compressed feasibility and create an admissible path $P$ from a free cluster $Z \in \mathcal{B}_{F}$ to a free cluster $Z'$ of $\mathcal{A}_{F}$ inside the associated graph $\mathcal{AG}_{\cell^*}$. Using a relatively straight-forward and standard argument very similar to that used in Lemma \ref{lem:syncproof}, one can show that these dual adjustments do not violate the compressed feasibility conditions. 

\paragraph{Projecting an Augmenting Path}
The dual adjustment ensures that there is some admissible augmenting path $P$ in $\mathcal{AG}_{\cell^*}$. We create an augmenting admissible augmenting path in $\res{\cell^*}$ from some free vertex $b \in Z$ to $a \in Z'$ as follows: For any internal edge $(U,V)$ in $P$,  we can recursively use  \sync\ (Corollary~\ref{cor:auguse}) to retrieve an admissible path $\Pi_{u',v',\cell^*}$ where $u' \in U$ and $v' \in V$.  We make $u'$ (resp. $v')$ the representative of $U$ (resp. $V$) and denote it by $r(U)$ (resp. $r(V)$). For every vertex $Y$ on the path $P$ that does not have a representative, we choose an arbitrary vertex $p \in Y$ as its representative, $p = r(Y)$. Note that $P$ cannot have any vertex with two internal edges incident on it.  Next, for any bridge edge $(x,y)$ in $P$, we show how to connect their representatives. Suppose the bridge edge $(x,y)$ is non-local edge. Then, we connect $r(x)$ and $r(y)$ directly by a non-local edge in $\res{\cell^*}$. Otherwise, suppose $(x,y)$ is a local bridge edge. In this case, if $r(x)$ is matched to $r(y)$, we simply add the matching edge between them. Otherwise, if $r(x)$ is matched to $x'$ and $r(y)$ is matched to $y'$, the edges $(r(x),x')$, $(x',y')$ and $(y', r(y))$ are all local and admissible. We add them the three edges in this order to connect $r(x)$ to $r(y)$. The resulting path obtained is a compact admissible path from a free vertex in $B_F$ to a free vertex in $A_F$ as desired.

Note that the input compressed feasible matching satisfied (J) and the dual weight of every free cluster $ v$ in $\mathcal{B}_{\cell^*}$ is $y_{\max}$.  The dual adjustments conducted by the \search\ procedure will not decrease the dual weights of any vertex $v \in \mathcal{B}_{\cell^*}$ and will not increase the dual weight of any vertex $v \in \mathcal{A}_{\cell^*}$.  Furthermore, each dual adjustment conducted by the \search\ procedure increases the dual weight  of all free clusters of $\mathcal{B}_{\cell^*}$  by $\ell$ which is the largest increase among all clusters. Therefore, the new dual weight of free clusters is $y_{\max}+\ell$ which is the largest among all vertices of $\mathcal{B}_{\cell^*}$. Finally, by definition, every free vertex cluster $v$ of $\mathcal{A}_{\cell^*}$ has $\ell_v \ge \ell$ and, therefore, $y(v)$ remains $0$. In conclusion, after the execution of \search\ procedure (J) continues to hold.

\paragraph{Efficiency of \search:}
Next, we bound the time taken by the \search\ procedure. First, we bound the time taken for the Dijkstra search over $\mathcal{AG}_{\cell^*}'$ during some phase $i$. The root cell $\cell^*$ has a diameter of $\BigOT(\mu_i^2)$, and each of its pieces have a diameter of $\BigOT(\mu_{\lfloor 2i/3 \rfloor}^2) = \BigOT(\mu_{i}^{4/3})$. Therefore, there are $\BigOT((\mu_i^2 / \mu_{i}^{4/3})^2) = \BigOT(\mu_{i}^{4/3}))$ pieces of $\cell^*$. Each piece contains $\BigOT(\mu_i^{2/3})$ vertices in $V_{\cell^*}$ and $\BigOT(\mu_i^{4/3})$ internal edges in $E_{\cell^*}$. The number of bridge edges in $E_{\cell^*}$ could be much higher, but we observe that, by using the WSPD, the bridge edges incident on every vertex of $V_{\cell}$ can be divided into only $\BigOT(1)$ groups where the edges of each group have the same net-cost and direction. A similar technique is used for the Hungarian search described in Section \ref{sec:algorithm}. Therefore, the Dijkstra search over $\mathcal{AG'}_{\cell^*}$ can be executed in time near-linear in the number of internal edges and vertices of $\mathcal{AG'}_{\cell^*}$, i.e., $\BigOT(\mu_i^{8/3})$ time. 

After executing the Dijkstra search over $\mathcal{AG'}_{\cell^*}$, the \search\ procedure executes the \sync\ procedure to produce an admissible augmenting path $P$ in $\res{M}$. During this process, \sync\ only needs to be executed once per affected cell $\cell \in \aff{P}$. From Lemma \ref{lem:syncefficiency}, each execution of \sync\ on a cell of level $j$ takes $\BigOT(\mu_j^3)$ time. Recall that \sync\ is not called on any cell with level higher than $\lfloor 2i/3 \rfloor$, i.e., the level of the pieces of $\cell^*$. Therefore, the total time taken by the executions of the \sync\ procedure can be expressed as:
\[\BigOT(\sum_{j=0}^{\lfloor 2i/3\rfloor}|\affj{P}{j}|\mu_j^3).\]
Combining this with the time taken by the Dijkstra search gives the following bound on the time taken by the \search\ procedure.
\[\BigOT(\mu_i^{8/3} + \sum_{j=0}^{\lfloor 2i/3\rfloor}|\affj{P}{j}|\mu_j^3).\]

\subsubsection{\augment\ Operation}
\label{subsec:augment}
The \augment\ procedure accepts an admissible augmenting path $P$ in $\res{M}$. It then augments $M$ along $P$, and updates the data structure accordingly. To augment $M$ along $P$, we set $M \leftarrow M \oplus P$ and perform very similar dual weight changes to those described in Section \ref{sec:algorithm}. For any edge $(a,b)$ that was non-local prior to augmentation and became local after augmentation, let $\cell$ be the least common ancestor of $a$ and $b$ in $Q$. If there is a local bridge edge $(X,Y) \in E_\cell$ prior to augmentation such that $a$ enters $X$ and $b$ enters $Y$ through augmentation, we simply set $y(a) \leftarrow y(X)$ and $y(b) \leftarrow y(Y)$. Otherwise, if no such local edge existed, we set $y(a) \leftarrow y(a) - \mu_{ab}^2$. Using similar arguments to those given in Section \ref{sec:algorithm}, it can be shown that this dual weight assignment only decreases the dual weights of $y(a)$ and $y(b)$. 

After augmenting along $P$, the data structure must perform updates to account for the changes to the matching. Recall that the set $\aff{P}$ of affected cells contains all non-root cells of $\mathcal{T}_{\cell^*}$ that contain at least one vertex of $P$. To update the data structure, the procedure executes the \construct\ procedure on all cells of $\aff{P}$, processing cells at lower layers of $Q$ first. 

\paragraph{Efficiency of \augment}
To bound the efficiency of the \augment\ procedure, we consider the most expensive portion, which is the time taken for the calls to the \construct\ procedure on all affected pieces. Consider an execution of \augment\ that produced an augmenting path $P$. Recall that, from Lemma \ref{lem:constructefficiency}, the time taken for a single call to \construct\ on a cell of level $j$ is $\BigOT(\mu_j^3)$, which matches the time taken for the calls to the \sync\ procedure during the execution of \search\ that generated $P$. Using an identical argument, we can conclude that the total time taken by \augment\ is:
\[\BigOT(\sum_{j=0}^{\lfloor 2i/3\rfloor}|\affj{P}{j}|\mu_j^3).\]

Next, we show that the \augment\ operation will not violate compressed feasibility. For any point $p$ on the augmenting path, all clusters that contain $p$ have the same dual weight as $p$; this follows from Corollary~\ref{cor:auguse} and the fact that $P$ was found by recursively applying \sync\ on an admissible path in $\mathcal{AG}_{\cell^*}$. As was the case in Section~\ref{sec:algorithm}, the \augment\ procedure only reduces the dual weights of vertices in $A\cup B$. The \init\ procedure, when applied at the ancestors, may reduce the dual weights of some clusters. Recollect that the dual updates are done so that the local edges that they participate in satisfy (C2). Reducing the the dual weight of any cluster of type $B$ or reducing (i.e., increasing the magnitude of) the dual weight of any cluster of type $A$ only increases the slack on non-local edges. As a result, the compressed feasibility conditions holds continue to hold.    

\subsection{Simple Augmenting Paths} When the algorithm executes the \search\ procedure to generate an admissible augmenting path $P$, it is important that this augmenting path is simple, having no self-intersections. In this section, we argue that all augmenting paths generated by the algorithm are simple. We begin by specifying a useful property of any cycle in $\res{\cell^*}$.
\begin{lemma}
    \label{lem:slacknetcostcycle}
    Let $C$ be any cycle in $\res{\cell^*}$. Then $s(C) = \phi(C)$.
\end{lemma}
\begin{proof}
    Observe that, 
    \[s(C) = \sum_{(u,v) \in C} \phi(u,v) - |y(u)| + |y(v)| = \phi(C) + \sum_{(u,v) \in C} - |y(u)| + |y(v)|.\]
    Since each vertex of $C$ occurs as the head and tail of exactly one edge of $C$, the net contribution of each dual weight to $s(C)$ is $0$. Therefore, for any alternating cycle $C$ in $\mathcal{AG}_{\cell}$, we have $s(C) = \phi(C)$.
\end{proof}
Note that any path $P$ returned by the \search\ procedure must contain at least one non-local edge. Therefore, to argue that the algorithm never produces an admissible cycle, it is sufficient to argue that the graph $\res{\cell^*}$ does not contain any cycles with both $0$ net-cost and at least one non-local edge.  It is worth noting that, since local edges are admissible, any cycle consisting solely of local edges of the same class is admissible. However, the \search\ procedure will never return such a cycle as part of a path.

At the beginning of the algorithm, all edges are non-local, and have a positive net-cost. Therefore, we can assume the claim holds initially. The only operation performed by the algorithm that changes net-costs in $\res{\cell^*}$ is augmentation. So, it is sufficient to argue that, if there were no $0$ net-cost cycles with a non-local edge prior to some augmentation, there are also no such cycles after augmentation. Lemma \ref{lem:slacknetcostcycle} implies that any cycle has zero net-cost iff it is admissible with respect to \emph{every} possible $Q$-feasible dual assignment. Therefore, it suffices to argue that every cycle with at least one non-local edge after augmentation is inadmissible with respect to any single $Q$-feasible dual assignment. 

The algorithm does not explicitly maintain a $Q$-feasible set of dual weights, but it does implicitly maintain an associated $Q$-feasible matching. Namely, the \genduals\ procedure accepts a compressed feasible matching as input and returns a $Q$-feasible set of dual weights $y(\cdot)$ for the points of $A_{\cell^*} \cup B_{\cell^*}$. Instead of generating this set of dual weights in its entirety, the algorithm only generates the dual weights that may change during augmentation, i.e., those along the augmenting path $P$. Furthermore, the augmenting path $P$ produced by the algorithm is admissible w.r.t. $y(\cdot)$  from Corollary \ref{cor:auguse}. We can describe a set of $Q$-feasible dual weights $y'(\cdot)$ after the augmentation; for any vertex $v$ not on $P$, $y'(v) = y(v)$, and for any vertex on $P$, the new dual weight $y'(v)$ is assigned explicitly by the \augment\ procedure. Since the \augment\ procedure only reduces dual weights, $y'(v)$ is a $Q$-feasible matching with respect to the matching $M'$ after augmentation. Therefore, we simply need to argue that there are no admissible cycles w.r.t. $y'(\cdot)$ in $\res{M'}$ that have at least one non-local edge. We argue this in the following lemma. 
\begin{lemma}
    Let $M, y(\cdot)$ be any $Q$-feasible matching such that $\res{M}$ does not contain any admissible cycles with at least one non-local edge, and let $P$ be an admissible compact augmenting path with respect to $y(\cdot)$. Consider the matching $M' = M \oplus P$ and the set of dual weights $y'(\cdot)$ assigned during augmentation. Then $\res{M'}$ does not contain any admissible cycles w.r.t. $y'(\cdot)$ with at least one non-local edge.
\end{lemma}
\begin{proof}
Assume for the sake of contradiction that $\res{M'}$ contains an admissible alternating cycle $C'$ with at least one non-local edge. Observe that any vertex that experiences a dual weight change during augmentation only has its dual weight strictly reduce. This causes all non-matching edges incident on it to accumulate a strictly positive slack with respect to $M', y'(\cdot)$. Since we assumed that $C'$ is admissible, $C'$ cannot contain any such non-matching edge, which implies that $C'$ does not contain a vertex that experienced a dual weight change.

Now, consider any edge $(u,v)$ shared between $C'$ and $P$ that is in $M'$ and was a non-local non-matching edge w.r.t. $M'$. Since $(u,v)$ was an admissible non-local edge, $y(u) + y(v) = d_Q(u,v) + \mu_{uv}$. After augmentation, $(u,v)$ is a feasible matching edge with $y'(u) + y'(v) = d_Q(u,v)$. Therefore, the dual weight of one of the endpoints of $(u,v)$ decreased. Since $C'$ cannot use any vertex that experienced a dual weight decrease, $C'$ cannot use any edge $(u,v)$ that is in $M'$ but was non-local w.r.t. $M$.

We conclude that any edge $(u,v)$ of $M'$ on $C'$ must have been a local non-matching edge prior to augmentation. Since $(u,v)$ was local in $M$, $u$ was matched to a vertex $v'$, $v$ was matched to a vertex $u'$, and there must have been another non-matching local edge directed from $u'$ to $v'$ in $\res{M}$. Therefore, $\res{M}$ contains an admissible path $P_{u,v}=\langle v, u', v', u\rangle$ from $v$ to  $u$ in $\res{M}$. 

Using this fact, we can craft an admissible cycle $C$ in $\res{M}$ as follows: For any edge $(u,v)$ on $C'$ that is not on $P$, we add $(u,v)$ to $C$. Since neither $u$ nor $v$ experienced a dual weight change, any such edge is admissible w.r.t. $\res{M}$. For any edge $(u,v)$ on $C'$ that is also on $P$, we add the edges of $P_{u,v}$ to $C$. Note that all edges of $P_{u,v}$ were admissible local edges in $\res{M}$. Thus, $C$ forms an admissible cycle in $\res{M}$. Furthermore, any non-local edge of $C'$ is also a non-local edge in $C$. Since we assumed that $C'$ contains at least one non-local edge, this contradicts the assumption that $\res{M}$ did not contain any admissible cycles with at least one non-local edge.
\end{proof}

\section{Transforming Input}
\label{sec:transform}

In this section, given any point sets $A', B'\subset \mathbb{R}^2$ of $n$ points, we generate point sets $A$ and $B$ with $n$ points each such that each point of $A'$ (resp $B'$) maps to a unique point of $A$ (resp. $B$) and:
\begin{itemize}
     \item[(A1)] Every point in $A \cup B$ has non-negative integer coordinates bounded by $\Delta =n^{\BigO(1)}$, 
    \item[(A2)] No pair of points $a,b$ where $a \in A$ and $b \in B$ are co-located, i.e., $\|a-b\| \ge 1$,
    \item[(A3)] The optimal matching of $A$ and $B$ has a cost of at most $\BigO(n/\eps^2)$, and,
    \item[(A4)] Any $\eps$-approximate matching of $A$ and $B$ corresponds to an $3\eps$-approximate matching of $A'$ and $B'$.
\end{itemize}
 We start by computing an $n^{\BigO(1)}$-approximation of the optimal matching cost. First, we compute a $2n^2$-approximate bottleneck matching $M_\mathcal{B}$ of $A',B'$ in $\BigO(n\log{n})$ time using the algorithm of~\cite{av_scg04}; see Lemma 2.2 in their paper. Let $\beta$ be the optimal bottleneck distance; then each edge of $M_{\mathcal{B}}$ has a length of at most $2n^2\beta$. Therefore, the cost of $M_{\mathcal{B}}$ under squared-Euclidean distance is at most $2n^5\beta^2$. On the other hand, at least one edge of the optimal squared-Euclidean matching $\Mopt$ must have a length that is at least $\beta^2$; otherwise, there is a smaller bottleneck matching distance than $\beta$. Therefore, the optimal squared Euclidean cost is at least $\beta^2$. We conclude that $\cost{M_{\mathcal{B}}} \leq 2n^5\cost{\Mopt}$. 

Next, let $\Gamma=\sqcost{M_\mathcal{B}}$.  Then for each integer $i$ such that $2^i \in [\Gamma / (2n^5),\Gamma]$, let $\gamma=2^i$. For at least one of these $\BigO(\log{n})$ values of $\gamma$, we will have
\begin{equation}
\label{eq:gamma1}
    \gamma \leq \sqcost{\Mopt} \leq 2\gamma.
\end{equation}
For a sufficiently large constant $c_1$, we can execute $c_1n^{5/4}\poly\{\log{n}, 1/\eps\}$ steps of the algorithm for each value of $\gamma$, and, out of all executions that terminate, choose the one whose generated matching has the smallest cost. Therefore, we can assume that our algorithm has a value of $\gamma$ that satisfies \eqref{eq:gamma1}.
We rescale the point set by dividing all coordinates by $\sqrt{\frac{\gamma\eps^2}{256n}}$.
Let $\hat{A}$ and $\hat{B}$ be the resulting scaled points from $A'$ and $B'$ respectively. Since the scaling was uniform, the optimal matching $\hat{M}_{OPT}$ with respect to the scaled points $\hat{A}, \hat{B}$ is also optimal with respect to the original point sets $A',B'$. Similarly, an $\eps$-approximate matching with respect to $\hat{A},\hat{B}$ is also an $\eps$-approximate matching with respect to the original points $A',B'$. As a result of the scaling, it is easy to see from \eqref{eq:gamma1} that the resulting optimal matching cost can be bounded by,
\begin{equation}
\label{eq:optscaledcost}
    256n/\eps^2 \leq \sqcost{\Mopth} \leq 512n/\eps^2.
\end{equation}
Next, we explain how to ensure that the diameter of the point set is polynomial in $n$. We construct a randomly shifted grid $\hat{G}$, where each cell in the grid has side-length $2048\cdot n^3/\eps^2$. Since each edge of the optimal matching has cost at most $512n/\eps^2$, each edge of $\hat{M}_\opt$ has a probability of at most $1/n^2$ of crossing between two different cells of $\hat{G}$. The probability that at least one of the $n$ optimal matching edges crosses between different cells is at most $1/n$. Therefore, we can split the point set using cells of $\hat{G}$, treating the points within each cell as a separate problem, and combine the resulting matchings together. With probability at least $1-1/n$, this splitting of points will not destroy any edges of the optimal matching. Therefore, we can assume that the points are non-negative coordinates bounded by $\Delta = n^{\BigO(1)}$.

Finally, we round the point sets $\hat{A},\hat{B}$ to integer coordinates $A$ and $B$. Given any location $p=(x,y)$ with integer coordinates, we say $p$ is \emph{even} if $x+y$ is even. Otherwise, $p$ is $\emph{odd}$. In order to ensure that no point of $A$ appears at the same location as a point of $B$, we round each point of $\hat{A}$ to the nearest even location and round each point of $\hat{B}$ to the nearest odd location. Note that, after this rounding, every edge $(a,b) \in A \times B$ has a length of at least $1$. The following Lemma proves that this rounding process distorts the cost of any matching by at most a $(1+\eps)$ factor. Noting that $(1+\eps)^2 \leq 3\eps$ implies (A4), completing the proof of properties (A1)--(A4). 
\begin{lemma}
For any point $\hat{a} \in \hat{A}$ (resp. $\hat{b} \in \hat{B}$) prior to rounding, let $a\in A$ (resp. $b \in B$) be the corresponding point after rounding. Let $\hat{M}$ be any matching with respect to $\hat{A}, \hat{B}$, and let $M = \bigcup_{(\hat{a},\hat{b}) \in \hat{M}}(a,b)$ be the corresponding matching with respect to the transformed points $A,B$. Then,
$$\sqcost{M} \leq (1+\eps)\sqcost{\hat{M}}.$$
\end{lemma}
\begin{proof}
First consider that the Euclidean length of any edge is distorted by at most $2$ from rounding. We have,
\begin{align*}
    \sqcost{M} &= \sum_{(a,b) \in M} \distsq{a}{b}\\
    &\leq \sum_{(\hat{a},\hat{b}) \in \hat{M}} \distsq{a}{b}\\
    &\leq \sum_{(\hat{a},\hat{b}) \in \hat{M}} (\disteuc{\hat{a}}{\hat{b}} + 2)^2\\
    &= \sum_{(\hat{a},\hat{b}) \in \hat{M}} (\distsq{\hat{a}}{\hat{b}} + 4\disteuc{\hat{a}}{\hat{b}} + 4)\\
    &=\sqcost{\hat{M}} + \sum_{(\hat{a},\hat{b}) \in \hat{M}} (4\disteuc{\hat{a}}{\hat{b}} + 4).
\end{align*}
We must show that $\sum_{(\hat{a},\hat{b}) \in \hat{M}} (4\disteuc{\hat{a}}{\hat{b}} + 4) \leq \eps\sqcost{\hat{M}}$.
Since $\sqcost{\hat{M}} \geq 256n/\eps^2$, it is sufficient to show that the quantity $\sum_{(\hat{a},\hat{b}) \in \hat{M}} (4\disteuc{\hat{a}}{\hat{b}} + 4)$ is at most $256n/\eps$, or that,  $\sum_{(\hat{a},\hat{b}) \in \hat{M}} \disteuc{\hat{a}}{\hat{b}} \leq 63n/\eps$. To bound this quantity, we divide the edges of $\hat{M}$ into two groups. First, consider that all the edges $(\hat{a},\hat{b}) \in \hat{M}$ with $\disteuc{\hat{a}}{\hat{b}} \leq 32 / \eps$ contribute a total value of at most $32n/\eps$. Next, consider the edges $(\hat{a},\hat{b}) \in \hat{M}$ with $\disteuc{\hat{a}}{\hat{b}} > 32/\eps$. For each such edge, we have, $\disteuc{\hat{a}}{\hat{b}} \leq \eps\distsq{\hat{a}}{\hat{b}}/32$. Since $\sqcost{\hat{M}} \leq 512n/\eps^2$, the total contribution from these edges is at most $16n^2 / \eps$. Thus, $\sum_{(\hat{a},\hat{b}) \in \hat{M}} \disteuc{a}{b} \leq 63n/\eps$, completing the proof.
\end{proof}
\ignore{See Comment}

\section{Quadtree Distance Proofs}
\label{sec:ommittedproofs}
\subsection{Proof of Lemma \ref{lem:subcellcount}}

For any cell $\cell$ of $Q$ with level $i$, the total number of subcells is $\BigOT(\mu_i)$. 

\begin{proof}
We bound the number of subcells by giving an upper bound on the number of leaves of $Q_{\cell}$. Let $p$ be the center of $\cell$ and let $j=  \lfloor i/2\rfloor - 2\log \frac{\log \Delta}{\eps} - c_1$ be such that $\mu_i=2^j$ is the minimum subcell size for $\cell$. Consider a set of concentric axis-parallel squares $S_1',\ldots, S_t'$, where each $S_r'$ is centered at $p$ with a side-length $2^i - (\frac{144}{\eps}+1)2^{j+r}$. Note that $t < i-j$. Consider all cells of $G_{j+r}$ that are completely contained inside $S_r'$. Let $S_r$ be the bounding square of these cells. Note that the distance $\ell_{\min}(S_r,\cell)$  is at least $\frac{144}{\eps}2^{j+r}$ and therefore, all subcells of $\subcells{\cell}$ inside the square $S_r$ (by condition (b) for subcell construction) are cells of $G_k$ for some $k \ge j+r$. The total number of subcells in the region $S_{r}\setminus S_{r+1}$ can be bounded by the maximum number of cells of $G_{j+r}$ that can fit inside this region. The side-length of $S_r$ is at least $2^i - (\frac{144}{\eps}+1)2^{j+r}$ and the side-length of $S_{r+1}$ is at most $2^i - (\frac{144}{\eps})2^{j+r+1}$. Therefore, the total number of cells of $G_{j+r}$ that can fit inside this region is $\BigOT(2^{i-j-r})$. The values for $r$ can range from $1$ to $i-j$. Therefore, the total number of subcells is at most $\BigOT(i 2^{i-j})= \BigOT(\mu_i)$. 
\end{proof}

\subsection{Proof of Lemma \ref{lem:distapprox}}
In the following lemma, we use the $Q$-feasibility conditions to upper bound the cost of any $Q$-optimal matching by $\sum_{(a,b) \in \Mopt}d_Q(a,b) + \mu_{ab}^2$. This will assist in proving Lemma \ref{lem:distapprox}.

\begin{lemma}
\label{lem:qopt}
For any $Q$-optimal matching $M$ and set of dual weights $y(\cdot)$ on the vertices of $A \cup B$, then 
$w(M) \le \sum_{(a,b) \in \Mopt}d_Q(a,b) + \mu_{ab}^2$.
\end{lemma}
\begin{proof}
 For any edge $(a,b)$ in the matching $M$, from equation~\eqref{eq:distlb}, $\distsq{a}{b} \le d_Q(a,b)$ and so, 
\begin{equation}
    \label{eq:opt}
    w(M) = \sum_{(a,b) \in M}\distsq{a}{b} \le \sum_{(a,b) \in M}d_Q(a,b).
\end{equation}

If $(a,b) \in M$, then $(a,b)$ is local and from~\eqref{eq:feas2} we have $y(a)+y(b) = d_Q(a,b).$
Since $M$ is a perfect matching, 
$$\sum_{(a,b)\in M} d_Q(a,b) = \sum_{(a,b) \in M} (y(a) + y(b)) = \sum_{v \in A\cup B} y(v).$$

Finally, consider the edges of the optimal matching $M_{\mathrm{OPT}}$. From the fact that $\Mopt$ is a perfect matching, and from the $Q$-feasibility conditions,
\begin{equation}
\label{eq:cost}
\sum_{v \in A \cup B} y(v) = \sum_{(a,b) \in \Mopt}y(a)+y(b)\le \sum_{(a,b) \in \Mopt} d_Q(a,b) + \mu_{ab}^2.
\end{equation}
Combining equations \eqref{eq:feas2}, \eqref{eq:opt}, and \eqref{eq:cost} completes the proof.
\end{proof}
To prove Lemma \ref{lem:distapprox}, we first need to show the following auxiliary claim.
\begin{lemma}
\label{lem:distaux}
For any two points $p,q \in A\cup B$, let $\square$ be the least common ancestor of $p$ and $q$ in $Q$, where $\square$ is a cell in $G_i$, and let $(\Psi_p, \Psi_q)$ be the WSPD pair in $\mathcal{W}_{\cell}$ that contains $p$ and $q$ respectively. 
\begin{itemize}
    \item[(i)] If $\disteuc{p}{q} \ge (144/\eps)\mu_i$, then $d_Q(p,q)+\mu_{pq}^2 \le (1+3\eps/8)\distsq{p}{q}$,
    \item[(ii)] If $\disteuc{p}{q} < (144/\eps)\mu_i$, then the subcells that contain $p$ and $q$, i.e., $\xi_p, \xi_q$,  have a side-length of $\mu_i$.
\end{itemize}
\end{lemma}
\begin{proof}
Let $j = \lfloor i/2\rfloor - 2\log \frac{\log \Delta}{\eps} - c_1$. Then the minimum subcell size $\mu_i$ is $2^j$. 
Let $\cell_1$ and $\cell_2$ be the two children of $\cell$ such that $p$ is inside $\cell_1$ and $q$ is inside $\cell_2$.
Let $t$ be an integer such that $\frac{144}{\eps}2^t \le \disteuc{p}{q} \le \frac{144}{\eps}2^{t+1}$. For (i), $t \ge j$ and $\xi_p$ (resp. $\xi_q$) is a cell of grid $G_k$ (resp. $G_{k'}$) such that $k \le t+1$ (resp. $k' \le t+1$). Let $(\Psi_p,\Psi_q)\in \wspd_{\cell}$ be the representative pair of $(p,q)$. The diameters of $\xi_p$ and $\xi_q$ are at most $\sqrt{2}\times2^{t+1} \le \frac{\eps}{36\sqrt{2}}\disteuc{p}{q}$, and, therefore, 
\begin{align*}
    (1+\eps/12)\ell_{\max}(\xi_p,\xi_q) &\le (1+\eps/12)(\disteuc{p}{q} + \sqrt{2}\times 2^{t+2}) \\
    &\le \disteuc{p}{q} (1 + \frac{\eps}{18\sqrt{2}})(1+\eps/12) \\
    &\le (1+\eps/12)(1+\eps/24)\disteuc{p}{q}\\ &\le (1+\eps/8)\disteuc{p}{q}.
\end{align*}
 Similarly, we can bound 
 \begin{align*}
     \mu_{pq}^2 \le 2^{2t+2} \le (\eps/72)^2\distsq{p}{q} \le (\eps/72)\distsq{p}{q}.
 \end{align*}
Combining the previous two bounds together with \eqref{eq:wspdguarantee} gives the following:
\begin{align*}
    d_Q(p,q) +\mu_{pq}^2&=\ell_{\max}(\Psi_p,\Psi_q)^2  +\mu_{pq}^2\\
    &\le (1+\eps/12)^2\ell_{\max}(\xi_p,\xi_q)^2 +\eps/72\distsq{p}{q}\\
    &\le (1+\eps/8)^2\distsq{p}{q} +\eps/72\distsq{p}{q}\\
    &\le (1+3\eps/8)\distsq{p}{q}.
\end{align*}

For (ii), observe that the distances of $\xi_p$ and $\xi_q$ to the boundaries of $\cell_1$ and $\cell_2$ respectively are less than $\frac{144}{\eps}\mu_i$. From the subcell construction procedure, $\xi_p$ and $\xi_q$ should be cells of the minimum subcell size $\mu_i$.
\end{proof}
Finally, in the following Lemma, we argue that, for any edge $(p,q)$ the expected value of the quadtree distance $d_Q(p,q)$ plus the additional additive error $\mu^2_{pq}$ is at most $(1+\eps)/2$ times the squared Euclidean distance $\distsq{p}{q}$. Combining this with Lemma~\ref{lem:qopt} and applying linearity of expectation immediately gives Lemma~\ref{lem:distapprox}.
\begin{lemma}
Given a randomly shifted quad tree $Q$, for any pair of points $(p,q) \in A \times B$, $$\expect{d_Q(p,q) + \mu_{pq}^2} \le (1+\eps/2)\distsq{p}{q}.$$
\end{lemma}
\begin{proof}
Let $\square \in G_i$ be the least common ancestor of $p$ and $q$ in $Q$, and let $\xi_p, \xi_q$ be the subcells that contain $p$ and $q$ respectively. Let $\mathbb{E}_i[d_Q(p,q)]$ be the expected value of of the distance given that $\square \in G_i$.
\begin{eqnarray*}
\expect{d_Q(p,q)} &=& \sum_{i=1}^{\log \Delta} \prob{\square \in G_i}\mathbb{E}_i[d_Q(p,q)+\mu_{pq}^2]\\
&\le& \sum_{i=1}^{\log \Delta} \prob{\square \in G_i}\biggl(\mathbb{E}_i \biggl[d_Q(p,q) + \mu_{pq}^2\mid \disteuc{p}{q} \ge (144/\eps)\mu_i\biggr]\\
& &+\mathbb{E}_i\biggl[d_Q(p,q)+\mu_{pq}^2\mid \disteuc{p}{q} < (144/\eps)\mu_i)\biggr]\biggr)\\
&\le& (1+3\eps/8)\distsq{p}{q} + \sum_{i=1}^{\log \Delta} \prob{\square \in G_i}\biggl(\mathbb{E}_i\biggl[d_Q(p,q)+\mu_{pq}^2\mid \disteuc{p}{q} < (144/\eps)\mu_i\biggr]\biggl).
\end{eqnarray*}
To complete the proof, we upper bound the second term of the RHS by $(\eps/8)\distsq{p}{q}$. First, note that 
$$\prob{\square\in G_i} \le \|p-q\|_1/2^{i-1} \le \distsq{p}{q}/2^{i-1}.$$
Since $\disteuc{p}{q} < (144/\eps)\mu_i$, by Lemma~\ref{lem:distaux}(ii), $\xi_p$ and $\xi_q$ have a side-length of the minimum subcell size $\mu_i$ and therefore, we can bound 
\begin{align*}
  d_Q(p,q) +\mu_{pq}^2&\le (1+\eps/12)\ell_{\max}(\xi_p,\xi_q)^2 +\mu_{pq}^2\\
  &\le (1+\eps/12)(\disteuc{p}{q}+2\sqrt{2}\mu_i)^2 +\mu_i^2 \\
  &\le (1+\eps/12)((144/\eps)\mu_i+2\sqrt{2}\mu_i)^2 +\mu_i^2  \\
  & \le 9000\mu_i^2/\eps^2.
\end{align*}
Recall that the minimum subcell size $\mu_i = 2^{\lfloor i/2\rfloor - 2\log \frac{\log \Delta}{\eps}-c_1}$, where $c_1 > 0$ is a constant. By setting $c_1$ to be sufficiently large, we get $\mu_i^2 \leq \eps^42^i/(16\log^4\Delta)$. Therefore, $d_Q(p,q) +\mu_{pq}^2 \leq \eps2^i/(16\log\Delta)$, and we finally have,
\begin{eqnarray*} \sum_{i=1}^{\log \Delta} \prob{\square \in G_i}\biggl(\mathbb{E}_i\biggl[d_Q(p,q) +\mu_i^2\mid \disteuc{p}{q} < (144/\eps)\mu_i\biggr]\biggl) 
&\le& \sum_{i=1}^{\log \Delta} \biggl((\distsq{p}{q}/2^{i-1}) \biggl(\frac{\eps}{16\log \Delta}\biggr)2^i\biggr) \\&\le& \eps/8 \distsq{p}{q},
\end{eqnarray*}
as desired.
\end{proof}
\subsection{Computing an $\eps$-Approximate Matching with High Probability}
\label{subsec:with-high-probability}
From Lemma \ref{lem:distapprox}, we have
\begin{equation}
    \expect{\sum_{(p,q) \in \Mopt} d_Q(p,q) + \mu_{pq}^2} \le (1+\eps/2)w(\Mopt).
    \label{eq:expected-total-dist}
\end{equation}
However, it is desirable to remove the need for expected values. Instead, we explain how to ensure that
\begin{equation}
    \sum_{(p,q) \in \Mopt} d_Q(p,q) + \mu_{pq}^2 \leq (1+\eps)w(\Mopt)
    \label{eq:distance-approx-no-expectation}
\end{equation}
with high probability. We can then design a $\BigOT(n^{5/4})$ time algorithm for computing a $Q$-optimal matching under the assumption that \eqref{eq:distance-approx-no-expectation} holds. To ensure this assumption, we can execute our algorithm $\log_2(n)$ times, and among all executions that terminate in $\BigOT(n^{5/4})$ time, use the one that produces the smallest cost. Within each execution, from \eqref{eq:distlb} we have
\[w(\Mopt) = \sum_{(p,q) \in \Mopt} \distsq{p}{q} \leq \sum_{(p,q) \in \Mopt} d_Q(p,q) + \mu_{pq}^2.\]

Furthermore, by combining this with \eqref{eq:expected-total-dist} we have that \eqref{eq:distance-approx-no-expectation} holds with probability at least $1/2$. Therefore, the probability that \eqref{eq:distance-approx-no-expectation} is satisfied by at least one of the $\log_2(n)$ random shifts is at least $1-1/n$. We present an $\BigOT(n^{5/4})$ time algorithm for computing a $Q$-optimal matching under the assumption that \eqref{eq:distance-approx-no-expectation} holds. By combining this assumption, with Lemma \ref{lem:qopt}, we have that any $Q$-optimal matching $\eps$-approximates the optimal RMS matching. Therefore, to prove Theorem \ref{theorem:main}, it is sufficient to give an $\BigOT(n^{5/4})$ time algorithm for computing a $Q$-optimal matching. 
\bibliographystyle{mystyle}
\bibliography{bib.bib}

\end{document}